\documentclass{toc}

\usepackage{amsfonts}
\usepackage{amsmath}
\usepackage{amssymb}
\usepackage{graphicx}
\usepackage{array}
\usepackage{multirow}
\usepackage{hyperref}
\usepackage{tabularx} 
\newcommand{\bra}[1]{\langle#1|}

\newcommand{\ket}[1]{|#1\rangle}

\newcommand{\av}[1]{\left| #1 \right|}

\newcommand{\ketbra}[2]{\ket{#1}\bra{#2}}
\newcommand{\ketbraa}[1]{\ketbra{#1}{#1}}

\newcommand{\np}{\textsf{NP}}

\let\originalleft\left
\let\originalright\right
\renewcommand{\left}{\mathopen{}\mathclose\bgroup\originalleft}
\renewcommand{\right}{\aftergroup\egroup\originalright}

\usepackage{color,graphicx}

\newcommand{\cls}[1]{\mathrm{#1}}

\newcommand{\ot}{\otimes}

\newcommand{\pa}[1]{(#1)}
\newcommand{\Pa}[1]{\left(#1\right)}

\newcommand{\set}[1]{\{#1\}}

\newcommand{\braket}[2]{  \langle #1 \vert #2 \rangle  }

\newcommand{\kb}[1]{\ket{#1} \bra{#1}}

\DeclareMathOperator{\trace}{Tr}
\newcommand{\ptr}[2]{\trace_{#1}\pa{#2}}
\newcommand{\Ptr}[2]{\trace_{#1}\Pa{#2}}

\newcommand{\tinyspace}{\mspace{1mu}}
\newcommand{\abs}[1]{|\tinyspace#1\tinyspace|}
\newcommand{\Abs}[1]{\left|\tinyspace#1\tinyspace\right|}
\newcommand{\norm}[1]{\lVert\tinyspace#1\tinyspace\rVert}
\newcommand{\Norm}[1]{\left\lVert\tinyspace#1\tinyspace\right\rVert}
\newcommand{\tnorm}[1]{\norm{#1}_1}
\newcommand{\Tnorm}[1]{\Norm{#1}_1}
\newcommand{\lnorm}[1]{\norm{#1}_{\owLOCC}}
\newcommand{\Lnorm}[1]{\Norm{#1}_{\owLOCC}}

\newcommand{\LOCC}{\operatorname{LOCC}}
\newcommand{\owLOCC}{{1\textnormal{-}\LOCC}}

\runningtitle{Quantum interactive proofs and the complexity of separability testing }

\tocpdftitle{Quantum interactive proofs and the complexity of separability testing}

\runningauthor{Gus Gutoski, Patrick Hayden, Kevin Milner, and Mark M.~Wilde}
\copyrightauthor{Gus Gutoski, Patrick Hayden, Kevin Milner, and Mark M.~Wilde}
\tocpdfauthor{Gus Gutoski, Patrick Hayden, Kevin Milner, Mark M.~Wilde}

\begin{document}

\begin{frontmatter}
\title{Quantum interactive proofs and the complexity of separability testing}


\author[gutoski]{Gus Gutoski\thanks{Supported by Government of Canada through Industry Canada, the Province of Ontario through the Ministry of Research and Innovation, NSERC, DTO-ARO, CIFAR, and QuantumWorks.}}
\author[hayden]{Patrick Hayden\thanks{Supported by Canada Research Chairs program, the Perimeter Institute, CIFAR, NSERC and ONR through grant N000140811249. The Perimeter Institute is supported by the Government of Canada through Industry Canada and by the Province of Ontario through the Ministry of Research and Innovation.}}
\author[milner]{Kevin Milner\thanks{Supported by NSERC.}}
\author[wilde]{Mark M.~Wilde\thanks{Supported by Centre de Recherches Math\'ematiques.}}

\begin{abstract}
  We identify a formal connection between physical problems related to the detection of separable (unentangled) quantum states and complexity classes in theoretical computer science.
  In particular, we show that to nearly every quantum interactive proof complexity class
  (including $\cls{BQP}$, $\cls{QMA}$, $\cls{QMA}(2)$, and $\cls{QSZK}$), there corresponds a natural separability testing problem that is complete for that class.
  Of particular interest is the fact that the problem of determining whether an isometry can be made to produce a separable state is either $\cls{QMA}$-complete or $\cls{QMA}(2)$-complete, depending upon whether the distance between quantum states is measured by the one-way LOCC norm or the trace norm.
  We obtain strong hardness results by proving that for each $n$-qubit maximally entangled state
   there exists a fixed one-way LOCC measurement that distinguishes it from any separable state with error probability that decays exponentially in $n$.

\end{abstract}


\tocacm{F.1.3}
\tocams{68Q10, 68Q15, 68Q17, 81P68}

\tockeywords{quantum entanglement, quantum complexity theory, quantum interactive proofs,
quantum statistical zero knowledge, BQP, QMA, QSZK, QIP, separability testing}

\end{frontmatter}

\section{Introduction}

Certain families of decision problems have proven to be particularly versatile and expressive in
complexity theory, in the sense that slightly varying their formulation can tune the difficulty of
the problems through a wide range of complexity classes. Adding quantifiers to the problem of
evaluating a Boolean formula, for example, brings the venerable satisfiability problem up through the
levels of the polynomial hierarchy \cite{Sto76} all the way up to PSPACE \cite{Sip96}, at
each level providing a decision problem complete for the associated complexity class.
Moreover, adding limitations to the format of the Boolean satisfiability problem gives decision problems complete for a
variety of more limited classes.\footnote{For example, it is known that if clauses of the Boolean satisfiability problem
are limited to two variables each, the resulting problem
(2SAT) is NL-complete \cite[Ch.~4.2, Theorem 16.3]{Pa94}, while if one allows only Horn clauses the
resulting problem
(HORNSAT) is P-complete \cite{CP10}, and if one removes any such limitations on
clauses the problem
(SAT) is NP-complete \cite{C71}.} Likewise, in the domain of
interactive proofs \cite{B85,GMR85,BM88,GMR89,W03,KW00,W09}, problems based on distinguishing probability
distributions or quantum states, depending on the setting, arise very naturally.

In the domain of quantum information theory, quantum mechanical entanglement is responsible for many
of the most surprising and, not coincidentally, useful potential
applications of quantum information \cite{HHHH09}, including quantum
teleportation \cite{PhysRevLett.70.1895},
super-dense coding \cite{PhysRevLett.69.2881}, enhanced communication capacities
\cite{BSST99, BSST02, CLMW10}, device-independent quantum key distribution \cite{Ekert:1991:661, VV12},
and communication complexity \cite{CB97}.
Thus,
deciding whether a given quantum state is separable (unentangled) or entangled is a prominent and
long-standing question that frequently resurfaces in different forms.
The complexity of determining whether a given mixed quantum
state is separable or entangled therefore arose early and was resolved: the problem is \np -complete with respect to
Cook reductions when the state
is specified as a density matrix and one demands an error tolerance no
smaller than an inverse polynomial in the dimension of the matrix~\cite{G03, G10}.

From a physics or engineering perspective, however, it is often more
natural to specify a quantum state as arising from a
sequence of specified operations (as in a quantum circuit) or the application of a local
Hamiltonian \cite{L96,berry2007efficient}.
This formulation of the quantum separability problem was studied by three of us \cite{HMW13,HMW14}, wherein it was shown that the problem is hard for both $\cls{QSZK}$ and $\cls{NP}$, even when one demands that no-instances be far from separable in one-way LOCC distance (and not merely in trace distance).
It was also shown that this one-way LOCC variant of the problem admits a two-message quantum interactive proof, putting it in $\cls{QIP}(2)$.
The exact complexity of this problem is still open.

In this paper, we explore several variations on the complexity of determining whether a state specified by a quantum circuit is separable or entangled, or whether all inputs to a channel specified by a quantum circuit lead to separable states at the channel output.
The properties we vary include the following:
\begin{enumerate}
\item Allowing arbitrary mixed states versus restricting attention to pure states.
\item Allowing arbitrary channels versus restricting attention to isometric channels.
\item We compare the difficulty of deciding whether entanglement is present (separable versus entangled states) with the difficulty of identifying any correlation whatsoever (product versus non-product states).
\item Measuring distance between quantum states using the trace norm or the so-called ``one-way LOCC norm'' of \cite{MWW09}.
\end{enumerate}
We study seven different combinations of these properties, obtaining problems that are complete for four different complexity classes based on quantum interactive proofs: $\cls{BQP}$, $\cls{QMA}$, $\cls{QMA}(2)$, and $\cls{QSZK}$.
Our study applies to multipartite states and channels, though only bipartite states and channels are required for the hardness results.
We obtain strong hardness results as a corollary of a new theorem establishing the existence of a fixed one-way LOCC measurement that successfully distinguishes a given $n$-qubit maximally entangled state from any separable state with error probability that decays exponentially in $n$.
(\expref{Theorem}{prop:owLOCC-to-sep} of \expref{Section}{sec:prelim:owLOCC-to-sep}.)

\textbf{Outline of paper.}
A detailed list of our complexity theoretic results is given in \expref{Figure}{fig:resultstable} of \expref{Section}{sec:overview}.
A summary of relevant concepts such as the one-way LOCC distance, various complexity classes,
the permutation and swap tests is given in \expref{Section}{sec:prelim}.
Our strong lower bound on the one-way LOCC distance between maximally entangled and separable states is proven in \expref{Section}{sec:prelim:owLOCC-to-sep}.
The completeness results are presented in \expref{Sections}{sec:bqp-complete}--\ref{sec:qszk-completeness}.
In \expref{Section}{sec:geometric-measure} we discuss how these completeness results provide operational interpretations for several geometric measures of entanglement discussed in \cite{WG03,CAH13} and references therein.
Finally, we conclude in \expref{Section}{sec:conclusion} with a summary of our results and a discussion of directions for future research.

%
%

\subsection{Overview of results}
\label{sec:overview}

\expref{Figure}{fig:resultstable} gives
a brief description of each problem and provides a concise summary of our results.
Below we give more details of our results along with their relation to prior results in the literature:
\begin{enumerate}

\item
  \textsc{Pure Product State} is $\cls{BQP}$-complete, as is the one-way LOCC version of the problem.
  (\expref{Theorem}{thm:bqp-completeness} of \expref{Section}{sec:bqp-complete}.)
  Membership in $\cls{BQP}$ follows from the soundness of the ``product test'' \cite{HM10}.
  Hardness of the one-way LOCC version follows from an application of \expref{Theorem}{prop:owLOCC-to-sep}.

\item
  The one-way LOCC version of \textsc{Separable Isometry Output} is $\cls{QMA}$-complete.
  (\expref{Theorem}{thm:qma-completeness} of \expref{Section}{sec:qma-complete}.)
  Membership in $\cls{QMA}$ follows from the existence of succinct $k$-extendible witnesses for separability \cite{BCY11}.
  (A similar approach was used in previous work by three of us to place the one-way LOCC version of \textsc{Separable State} inside $\cls{QIP}(2)$ \cite{HMW13,HMW14}.)
  Hardness follows from another application of \expref{Theorem}{prop:owLOCC-to-sep}.

\item
  \textsc{Pure Product Isometry Output}, \textsc{Product Isometry Output}, and \textsc{Separable Isometry Output} are $\cls{QMA}(2)$-complete.
  (\expref{Theorem}{thm:qma2-complete} and \expref{Corollary}{cor:qma2-complete} of \expref{Section}{sec:qma2-complete}.)
  Membership of \textsc{Pure Product Isometry Output} in $\cls{QMA}(2)$ follows from a simple application of the swap test combined with the collapse $\cls{QMA}(k)=\cls{QMA}(2)$ \cite{HM10}.
  Hardness is the result of a novel circuit construction (\expref{Figure}{fig:qma2-reduction}).
  Completeness for the other two problems follows by an equivalence to \textsc{Pure Product Isometry} (\expref{Section}{sec:qma2-complete:equivalence}).

\item
  \textsc{Product State} is $\cls{QSZK}$-complete.
  (\expref{Theorem}{thm:qszk-completeness} of \expref{Section}{sec:qszk-completeness}.)
  The result follows by an equivalence with the $\cls{QSZK}$-complete problem \textsc{Quantum State Similarity} \cite{W02,W09zkqa}.

\item
  The one-way LOCC version of \textsc{Separable State} is in $\cls{SQG}$, a competing-provers class known to coincide with $\cls{PSPACE}$ \cite{GW13}.
  (\expref{Proposition}{thm:co-QSEP-in-SQG} of \expref{Section}{sec:sqg-inclusion}.)
  As mentioned previously, this problem is already known to be contained in $\cls{QIP}(2)$ \cite{HMW13,HMW14}, which is a subset of $\cls{PSPACE}$ \cite{JUW09,JJUW11}.
  Thus, this new bound is not a complexity-theoretic improvement over prior work.

  However, it is interesting that this problem admits a succinct quantum witness, provided that the verifier is granted the additional ability to query a second, competing prover in his effort to check the veracity of the first prover's purported witness.
  By contrast, the two-message single-prover quantum interactive proof of \cite{HMW13,HMW14} depends critically upon the ability of the verifier to exchange \emph{two messages} with the prover.

\end{enumerate}

\begin{figure}
\begin{center}
\begin{tabular}{| m{3.8cm} | m{4.5cm} | m{2.7cm} | m{3cm} |}
  \hline\hline
  {\bf Problem name} & {\bf Description} & {\bf Complexity} & {\bf Circuit}
  \\ \hline\hline
  {\small $\bullet$ \textsc{Pure Product State}\newline $\bullet$ \textsc{Pure Product State}, one-way LOCC version}
    & {\small Is the state generated by the pure-state quantum circuit close to a product state?}
    & {\small $\cls{BQP}$-complete}
    & \hfill\includegraphics[scale=.4]{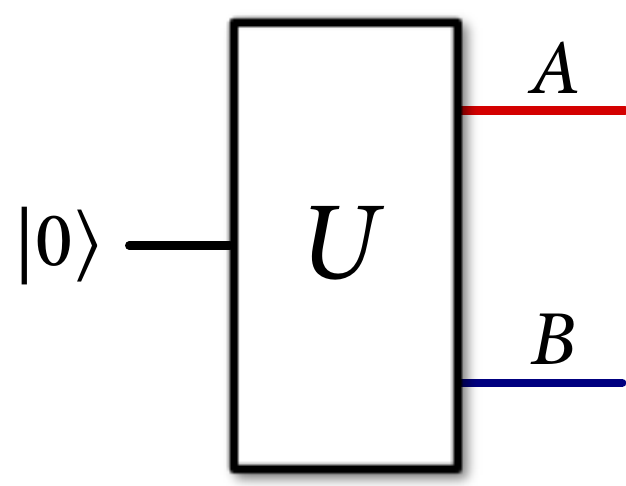}
  \\ \hline
  {\small $\bullet$ \textsc{Separable Isometry Output}, one-way LOCC version}
    & {\small Is there an input to the isometry such that the output is close to a separable state in trace distance, or does every input lead to an output that is far from separable in one-way LOCC distance?}
    & {\small $\cls{QMA}$-complete}
    & \hfill\multirow{2}{*}{\includegraphics[scale=.4]{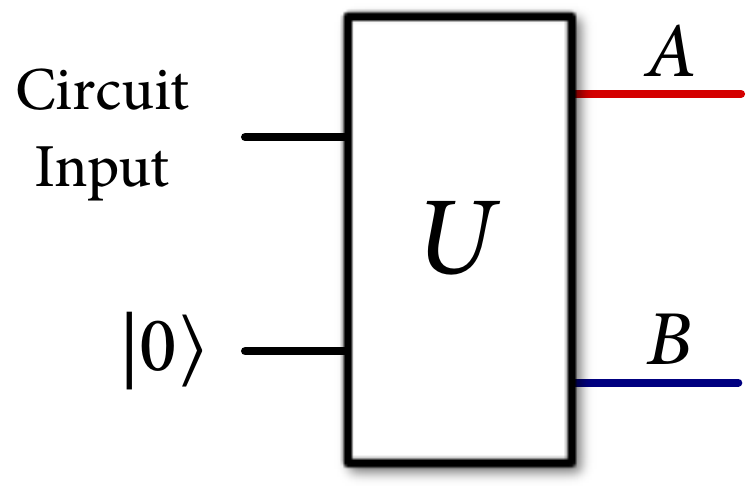}}
  \\ \cline{1-3}
  {\small $\bullet$ \textsc{Pure Product \mbox{Isometry} Output}}\par{\small $\bullet$ \textsc{Product Isometry Output}}\par{\small $\bullet$ \textsc{Separable Isometry Output}}
    & {\small Is there an input to the isometry such that the output is close to a product/separable state?}
    & {\small $\cls{QMA}(2)$-complete} &
  \\ \hline
  {\small $\bullet$ \textsc{Product State}}
    & {\small Is the state generated by the mixed-state circuit close to a product state?}
    & {\small $\cls{QSZK}$-complete}
    & \hfill\multirow{2}{*}{\includegraphics[scale=.4]{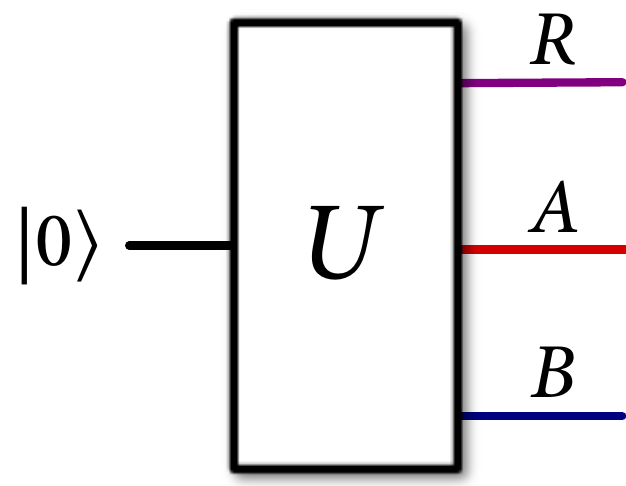}}
  \\ \cline{1-3}
  {\small $\bullet$ \textsc{Separable State}, one-way LOCC version}
    & {\small Is the state generated by the mixed-state circuit close to a \mbox{separable} state?}
    & {\small In $\cls{QIP}(2)$.\newline $\cls{QSZK}$-hard,\newline $\cls{NP}$-hard}.\newline \cite{HMW13,HMW14} &
  \\ \hline
  {\small $\bullet$ \textsc{Separable Channel Output}, one-way LOCC version}
    & {\small Is there an input to the channel such that the output is close to a separable state in trace distance or does every input lead to an output that is far from separable in one-way LOCC distance?}
    & {\small $\cls{QIP}$-complete\newline \cite{HMW13,HMW14}}
    & \includegraphics[scale=.4]{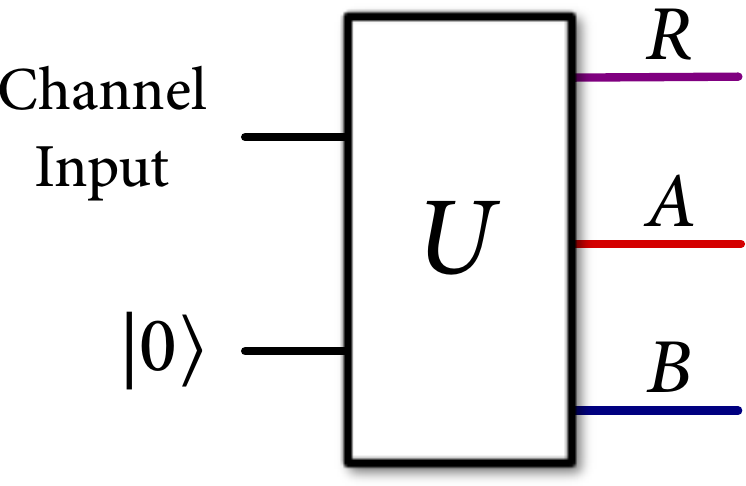}
  \\ \hline
\end{tabular}
\caption{
  The collected results of separability testing problems and their complexity.
  A ``one-way LOCC version'' of a problem means that distances for yes-instances are measured by the trace norm, but distances for no-instances are measured by the one-way LOCC norm.}
\label{fig:resultstable}
\end{center}
\end{figure}

\section{Preliminaries}
\label{sec:prelim}

This section summarizes some facts about quantum information and complexity theory
relevant for the rest of the paper.
Familiarity with both fields of study is assumed; our primary goal here is to establish notation and terminology.
Some references giving background on these topics are \cite{NC00,W11,W13} and \cite{W09,A13}.

\subsection{Registers, states, separable states}

A \emph{register} is a finite-level quantum system, which is implicitly identified with a finite-dimensional complex Euclidean space.
Registers are denoted with Roman capital letters $A,B,\dots$.
The \emph{state} of a register is described by a \emph{density matrix}, which is a positive semidefinite matrix $\rho$ with $\ptr{}{\rho}=1$.
A \emph{pure state} is a rank-one density matrix.
Pure states can be written in standard bra-ket notation $\rho=\kb{\psi}$ for some unit vector $\ket{\psi}$.
It is common practice to refer to unit vectors $\ket{\psi}$ as pure states.
The Greek letters $\phi,\psi$ are reserved for pure states and we often abbreviate $\kb{\psi}$ to $\psi$.

A multipartite state $\rho_{A_1\cdots A_l} $ is a \emph{product state}
if $\rho_{A_1\cdots A_l}=\rho_{A_1}\ot \cdots \ot \rho_{A_l}$
for states $\rho_{A_1}, \ldots , \rho_{A_l}$ of registers $A_1,
\ldots, A_l$, respectively.
A state is \emph{separable} if it can be written as a probabilistic mixture of product states
\cite{W89}.
That is, a multipartite state $\rho_{A_1\cdots A_l}$
is said to be \emph{separable} if it admits a decomposition of the following form:
\begin{equation}
\rho_{A_1\cdots A_l} = \sum_{y \in \mathcal{Y}} p_Y(y) \,
\sigma^{1,y}_{A_1} \otimes \cdots \otimes \sigma^{l, y}_{A_l} \, , \label{eq:def-sep-state}
\end{equation}
for collections $\{\sigma^{1,y}_{A_1}\}, \ldots, \{\sigma^{l,y}_{A_l}\}$
of quantum states and some probability
distribution $p_Y(y)$ over an alphabet $\mathcal{Y}$ \cite{W89}. By applying the spectral theorem
to each density operator, we can always find a decomposition of any
separable state in terms of pure product states:
\begin{equation}
\rho_{A_1\cdots A_l} = \sum_{z \in \mathcal{Z}} p_Z(z) \,
\vert \psi ^{1,z}\rangle\langle\psi^{1,z}\vert_{A_1} \otimes \cdots \otimes
\vert \psi^{l, z} \rangle \langle\psi^{l, z} \vert_{A_l} \, .
\label{eq:decomposition-into-pure-product-states}
\end{equation}
A state is \emph{entangled}
if it is not separable.

In the multipartite case, it is often necessary to specify the \emph{cut} or \emph{partition} of the registers relative to which $\rho$ is product or separable.
For example, a state $\rho$ of registers $ABCD$ could be a bipartite product state with respect to the cut $AB:CD$, yet it may not be a product state with respect to the tripartite cut $A:B:CD$ or the bipartite cut $AC:BD$.
We let $\mathcal{S}$ denote the set of all separable states with respect to a given cut.
Whenever the cut is not immediately clear from the context, we make it explicit
with an argument---for example, $\mathcal{S}(A:B:CD)$.

\subsection{Trace distance, fidelity}
\label{sec:trace-dist}
The \emph{Schatten 1-norm} $\tnorm{X}$ of a matrix $X$ is defined as the sum of the singular values of $X$.
(Hereafter we refer to this norm as simply the \emph{1-norm}.
This norm is sometimes called the \emph{trace norm} and is alternately denoted $\norm{X}_{\trace}$.)
The 1-norm characterizes the physically observable difference between two quantum states $\rho,\xi$ in the following sense: given a quantum register prepared in one of $\set{\rho,\xi}$ chosen uniformly at random, the maximum probability with which one can correctly identify the given state by a two-outcome measurement of that register is equal to $1/2 + \tnorm{\rho-\xi}/4$.
The measurement that achieves this maximal probability is known as the
\emph{Helstrom measurement}~\cite{H69}.

The quantity $\tnorm{\rho-\xi}$ is sometimes called the \emph{trace distance} between $\rho,\xi$.
The trace distance between two quantum states $\rho,\xi$ is given by the following variational characterization:
\begin{equation}
  \tnorm{\rho-\xi} = 2\max_{0\preceq\Pi\preceq I} \ptr{}{\Pi(\rho-\xi)} \label{eq:var-tr-dist} \, ,
\end{equation}
where the maximizing $\Pi^\star$ leads to the Helstrom measurement $\set{\Pi^\star,I-\Pi^\star}$.
A straightforward consequence of \eqref{eq:var-tr-dist} is that if two states are close in trace distance then they must have similar measurement statistics.
In particular, for all measurement operators $0\preceq\Pi\preceq I$ it holds that
\begin{equation}\label{eq:trace-inequality}
  \ptr{}{\Pi\rho} \geq \ptr{}{\Pi\xi} - \frac{1}{2} \tnorm{\rho - \xi} \, .
\end{equation}

The trace distance $\tnorm{\psi-\phi}$ between two pure states $\ket{\psi},\ket{\phi}$ is related to the inner product $\braket{\psi}{\phi}$ by the formula
\begin{equation} \label{eq:inner-product-to-tnorm}
  \abs{\braket{\phi}{\psi}}^2 = 1 - \tnorm{\psi-\phi}^2/4 \, .
\end{equation}
The following implication holds for any pure states $\phi,\psi$ and any $\varepsilon\in[0,1]$:
\begin{equation}
  \abs{\braket{\phi}{\psi}}^2\geq 1-\varepsilon \implies \tnorm{\phi-\psi}\leq 2\sqrt{\varepsilon} \, .
\end{equation}

The \emph{fidelity} is a pseudodistance measure for quantum states given by
\begin{equation} F(\rho,\xi) = \Tnorm{\sqrt{\rho}\sqrt{\xi}}^2 \end{equation}
for all density matrices $\rho,\xi$.
\emph{Uhlmann's Theorem} asserts that the fidelity between two states $\rho,\xi$ is the optimal squared overlap between purifications of $\rho,\xi$:
\begin{equation} F(\rho,\xi) = \max_{\ket{\phi_\rho}, \ket{\phi_\sigma}} \abs{\braket{\phi_\rho}{\phi_\sigma}}^2 \, . \end{equation}
Uhlmann's Theorem gives the fidelity an operational interpretation as the maximum probability with which a purification of $\rho$ would pass a test for being a purification of $\sigma$.
The fidelity and trace distance are related by the Fuchs-van-de-Graaf inequalities \cite{FvG99}:
\begin{equation}\label{eq:FvG-ineqs}
 1-\sqrt{F\Pa{\rho,\xi}} \leq \frac{1}{2}\Tnorm{\rho-\xi} \leq \sqrt{1-F\Pa{\rho,\xi}} \, .
\end{equation}

\subsection{Permutation test, swap test}

The \emph{permutation test} is a quantum circuit applied to a multi-register system $A_1,\dots,A_n$ with the property that the probability of passing is equal to the shadow of the state on
the symmetric subspace of the complex Euclidean space associated with $A_1,\dots,A_n$ (i.e., $\text{Tr}(\Pi^{\text{sym}}_{A_1,\dots,A_n} \rho )$)  \cite{K95} (see also \cite{BBDEJM97,KNY08}).
Furthermore, if the test passes, then the resulting state of those registers is supported on the symmetric subspace.
The test consists of the following steps:
\begin{enumerate}
\item Prepare an $n!$-dimensional ancillary register $W$ in a uniform superposition of all $n!$ computational basis states.
 (This is accomplished by applying the quantum Fourier transform to the
 all-zeros state $\ket{0}$ of $W$.)
\item Apply a controlled-permutation unitary that permutes registers $A_1,\dots,A_n$ according to the permutation indexed in register $W$.
\item Invert the quantum Fourier transform on $W$ and measure that register in the computational basis.
  Accept if and only if the measurement outcome is all zeros.
\end{enumerate}

A special case of the permutation test for $n=2$ is known as the \emph{swap test} \cite{BuhrmanC+01}.
(In this case the ancillary register $W$ is just a single qubit and the quantum Fourier transform is just the standard Hadamard gate.)
The swap test has the powerful property that if registers $A_1A_2$ are prepared in a pure product state $\ket{\psi}\ket{\phi}$ then the swap test passes with probability
\begin{equation}
  \frac{1}{2} + \frac{1}{2}\abs{\braket{\psi}{\phi}}^2 = 1 - \frac{1}{8}\Tnorm{\psi-\phi}^2 \, .
\end{equation}
Thus, with repetition, the swap test can be used to estimate the distance between any two unknown pure states.

\subsection{One-way LOCC distance}
\label{sec:1-locc-dist}
In this paper we are sometimes interested in the distinguishability of multipartite quantum states under the restriction that the distinguishing measurement must be implementable by local operations with unidirectional classical communication.
This class of measurements induces a matrix norm called the \emph{one-way LOCC norm} \cite{MWW09}.
For each matrix $X$ acting on the complex Euclidean space associated with registers $AB$, the one-way LOCC norm $\lnorm{X}$ of $X$ is defined by
\begin{equation} \lnorm{X} = \max_{\Lambda_{B\rightarrow M}} \tnorm{(I_A \ot \Lambda_{B \rightarrow M})(X)} \, , \end{equation}
where the maximization is over all \emph{quantum-to-classical} channels $\Lambda_{B \rightarrow M}$.
These are the channels that measure the contents of register $B$ and store the classical outcome in a new register $M$.
Every such channel has the form
\begin{equation} \Lambda_{B \rightarrow M}( \rho) = \sum_m \ptr{}{\Lambda_m\rho}\ketbra{m}{m} \, ,\end{equation}
where $\set{\ket{m}}$ is an orthonormal basis and $\set{\Lambda_m}$ forms a quantum measurement, meaning that each $\Lambda_m$ is positive semidefinite and $\sum_m\Lambda_m=I$.

This definition of the one-way LOCC norm is asymmetric: one could define another norm as a maximization over measurements of register $A$, and these norms are distinct.
It is clear from the definition that
\begin{equation}
\lnorm{X} \leq \tnorm{X} \, ,
\end{equation}
because the one-way LOCC measurements are a subset of all measurements.

The one-way LOCC norm extends naturally to multi-register systems \cite{LW12,BC12,BH12}.
In particular, for each matrix $X$ acting on the complex Euclidean space associated with registers $A_1\cdots A_l$, the $l$-partite one-way LOCC norm of $X$ is given by
\begin{equation}\label{eq:multi-1-LOCC}
  \lnorm{X} = \max_{\Lambda_{A_2},\dots,\Lambda_{A_l}} \Tnorm{\Pa{I_{A_1} \ot \Lambda_{A_2} \ot\cdots\ot \Lambda_{A_l}}(X)} \, ,
\end{equation}
where the maximization is now over quantum-to-classical channels $\Lambda_{A_2},\dots,\Lambda_{A_l}$.
The interpretation here when $X$ is a difference of two density matrices is that the last $l-1$ parties each perform a local measurement on their systems and communicate the results to the first party, who then attempts to distinguish the two states.

\subsection{Quantum interactive proofs}
\label{sec:qip-hierarchy}

A \emph{quantum interactive proof} consists of a conversation between a
polynomial-time quantum \emph{verifier} and a computationally unbounded
quantum \emph{prover} regarding some binary input string $x$.
The prover attempts to convince the verifier to accept $x$ and the verifier attempts to judge the veracity of the prover's argument.
A promise problem $L$ is said to admit a quantum interactive proof with \emph{completeness} $c$ and \emph{soundness} $s$ if there exists $c,s \in [0,1]$ such that $c>s$ and a verifier who meets the following conditions:
\begin{description}

\item[Completeness condition.]
If $x$ is a yes-instance of $L$, then the prover can convince the verifier to accept with probability at least $c$.

\item[Soundness condition.]
If $x$ is a no-instance of $L$, then no prover can convince the verifier to accept with probability higher than $s$.

\end{description}
The completeness and soundness parameters $c,s$ need not be fixed constants but may instead vary as a function of the input length $|x|$.
If these parameters are not specified then it is assumed that $L$ admits a quantum interactive proof for some choice of $c(|x|),s(|x|)$ for which there exists a polynomial-bounded function $p(|x|)$ such that $c-s\geq 1/p$.
The complexity class $\cls{QIP}$ consists of all promise problems that admit quantum interactive proofs and is known to coincide with $\cls{PSPACE}$ \cite{JJUW11}.

Often in the study of interactive proofs the precise values of $c,s$ are immaterial because error-reduction procedures can be used to transform any verifier for which $c-s\geq 1/p$ into another verifier for which $c$ tends toward one and $s$ tends toward zero exponentially quickly in the bit length of $x$.
(For example, sequential repetition followed by a majority vote can be used to reduce error for $\cls{QIP}$.)
For this reason, it is typical to assume without loss of generality that $c,s$ are constants such as $2/3,1/3$ or that $c$ is exponentially close to one and $s$ is exponentially close to zero whenever it is convenient to do so.
However, it is not always clear that a given complexity class is robust with respect to the choice of $c,s$ so it is good practice to be as inclusive as possible when defining these classes.

Interesting subclasses of $\cls{QIP}$ are obtained by restricting the number of messages in the interaction between the verifier and prover.
For each positive integer $m$, the class $\cls{QIP}(m)$ consists of those problems that admit a quantum interactive proof in which the verifier exchanges no more than $m$ messages with the prover.
It is known that three messages suffice for any quantum interactive proof, so that $\cls{QIP}=\cls{QIP}(3)$ \cite{KW00}, leaving a hierarchy of four classes defined by quantum interactive proofs.
Fundamental complexity classes such as $\cls{BQP}$ and $\cls{QMA}$ can be written in this notation as $\cls{QIP}(0)$ and $\cls{QIP}(1)$, respectively.
This hierarchy, along with other complexity classes considered in this paper, is depicted in \expref{Figure}{fig:hierarchy}.

\begin{figure}
\center
\includegraphics[scale=.19]{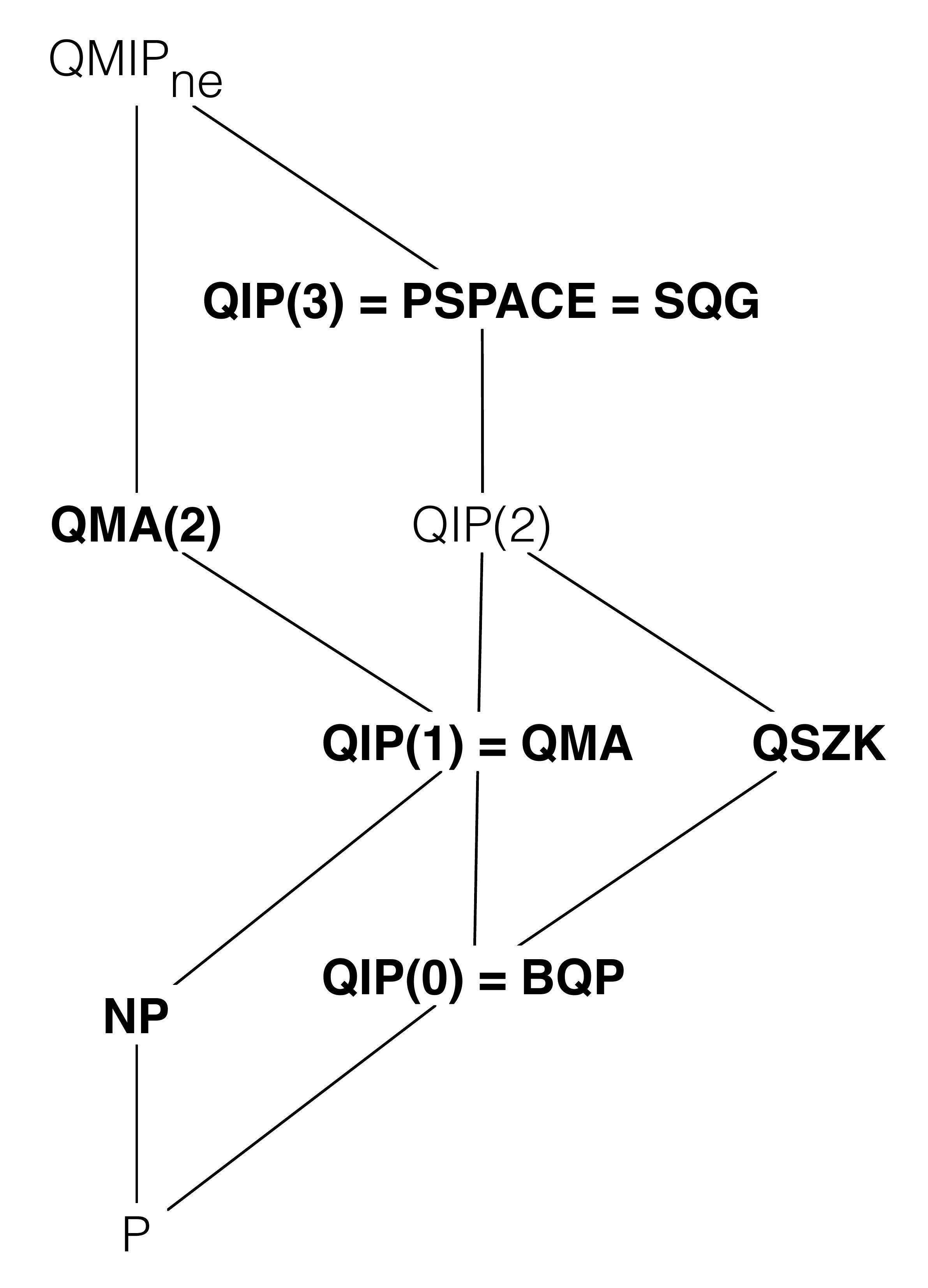}
\caption{\label{fig:hierarchy}
  The quantum interactive proof hierarchy and related classes discussed in this paper.
  A line denotes inclusion of the lower class in the higher class.
  (For example, $\cls{P}$ is a subset of $\cls{NP}$.)
  Classes for which a separability testing problem is known to be complete are in bold type.}
\end{figure}

A quantum interactive proof for a promise problem $L$ is said to be \emph{statistical zero knowledge} if for each yes-instance $x$ of $L$ the verifier learns nothing from the prover beyond the veracity of the claim ``$x$ is a yes-instance of $L$''.
This property is formalized via a simulation-based definition of ``knowledge.''
The complexity class of promise problems that admit statistical zero knowledge quantum interactive proofs is called $\cls{QSZK}$.
We need not concern ourselves with a precise definition of this class, since our completeness results are established by equivalence to another $\cls{QSZK}$-complete problem.
The reader is referred to the seminal works in \cite{W02,W09zkqa} for more information on statistical zero knowledge quantum interactive proofs.

Other interesting variations of the quantum interactive proof model are obtained by considering multiple cooperating or competing provers.
For example, one can consider a variant of $\cls{QMA}$ in which $k$ distinct and unentangled provers cooperate in order to convince the verifier to accept.
The resulting complexity class is called $\cls{QMA}(k)$ and is known to satisfy $\cls{QMA}(k)=\cls{QMA}(2)$ for all integers $k\geq 2$ \cite{HM10}.
The only known bounds for $\cls{QMA}(2)$ are the trivial bounds $\cls{QMA}\subseteq\cls{QMA}(2)\subseteq\cls{NEXP}$.
Evidence that $\cls{QMA}(2)$ is strictly larger than $\cls{QMA}$ was presented in Refs.~\cite{ABDFS08,BT09,B10,GNN12,CF13}.

Despite the lack of any decent upper bound on $\cls{QMA}(2)$, we are aware of only two problems in $\cls{QMA}(2)$ that are not also known to be in $\cls{QMA}$: the pure-state $N$-representability problem \cite{LCV07} and the separable sparse Hamiltonian problem \cite{CS12}.
Of these, only the latter is known to be $\cls{QMA}(2)$-complete.
The present paper gives another $\cls{QMA}(2)$-complete problem in \expref{Section}{sec:qma2-complete}.

Alternately, one could consider quantum interactive proofs with two competing provers: one prover---the \emph{yes-prover}---tries to convince the verifier to accept $x$ while the other prover---the \emph{no-prover}---tries to convince the verifier to reject $x$.
As before, interesting complexity classes are obtained by restricting the number and timing of messages in the interaction between the verifier and provers.
In \expref{Section}{sec:sqg-inclusion} we exhibit a protocol in which the verifier receives a single message from the yes-prover and then exchanges two messages with the no-prover.
The complexity class of promise problems that admit such proofs is called $\cls{SQG}$ (for ``short quantum games'') and is known to coincide with $\cls{PSPACE}$ \cite{GW13}.

Each of the aforementioned complexity classes is known to be robust with respect to the choice of completeness and soundness parameters $c,s$, meaning that any protocol for which $c$ is larger than $s$ plus an inverse polynomial in the input length can be amplified into a new protocol with $c$ exponentially close to one and $s$ exponentially close to zero.
Error reduction for $\cls{BQP}$ follows immediately from Chernoff-type bounds via sequential repetition followed by a majority vote.
Error-reduction results for $\cls{QIP}$, $\cls{QIP}(2)$, $\cls{QMA}$, $\cls{QSZK}$, $\cls{QMA}(2)$, and $\cls{SQG}$ were established in \cite{KW00,JUW09,MW05,W02,HM10,GW13},
respectively.\footnote{
  That $\cls{SQG}$ is robust with respect to error follows from the containments $\cls{SQG}(c,s)\subseteq \cls{PSPACE}$ for any $c-s>1/\mathrm{poly}$ \cite{GW13} and $\cls{PSPACE}\subseteq \cls{SQG}(1-\epsilon,\epsilon)$ for any desired exponentially small $\epsilon$ \cite{GW05}.
  However, the ``error reduction procedure'' induced here is very circuitous:
  a high-error short quantum game must be simulated in polynomial space, and then that polynomial-space computation must be converted back into a low-error short quantum game via proofs of $\cls{IP}=\cls{PSPACE}$ \cite{LFKN92,S92}.
  It is not known whether a more straightforward transformation such as parallel repetition followed by a majority vote could be used to reduce error for $\cls{SQG}$.}

\section{One-way LOCC distance to a separable state}
\label{sec:prelim:owLOCC-to-sep}

In this section we prove a theorem that enables us to establish strong hardness results for various separability testing problems appearing later in the paper.

If $\ket{\phi}$ is any maximally entangled pure state of two $n$-qubit registers $AB$ then
\begin{equation} \label{eq:fidelity-to-sep}
  \max_{\sigma\in\mathcal{S}(A:B)} F(\phi,\sigma) = 2^{-n} \, .
\end{equation}
A concise proof of the above equality can be found in \cite[Lecture 17]{W04}. Applying the above relation and \eqref{eq:FvG-ineqs}, we find the following result for the trace distance:
\begin{equation}
\min_{\sigma\in\mathcal{S}(A:B)} \left\Vert \phi - \sigma \right\Vert_1 \geq 2\left(1-2^{-2n}\right) \, .
\end{equation}

In fact, a much stronger statement holds:
every maximally entangled state is exponentially far from separable not only in trace distance, but also in one-way LOCC distance.
It appears that this observation has not yet been made explicitly in the literature, so we provide a proof. (However, we note that it is certainly implicit in many places in the literature.)

\begin{theorem}[Minimum one-way LOCC distance to separable]
\label{prop:owLOCC-to-sep}

  For all maximally entangled pure states $\ket{\phi}$ of two $n$-qubit registers $AB$ it holds that
  \begin{equation} \min_{\sigma\in\mathcal{S}(A:B)} \Lnorm{\phi-\sigma}\geq 2(1-(2/3)^{n}) \, .   \end{equation}
  Moreover, this bound is witnessed by a \emph{fixed} one-way LOCC measurement that depends only on $\ket{\phi}$.

\end{theorem}

\begin{proof}
Let $A^{n}\equiv A_{1}\cdots A_{n}$ denote Alice's $n$ qubits, and let
$B^{n}\equiv B_{1}\cdots B_{n}$ denote Bob's. By the local unitary equivalence
of maximally entangled states, it suffices to exhibit a fixed one-way LOCC
measurement that successfully distinguishes any separable state $\sigma
_{A^{n}:B^{n}}$ from $n$ singlets%
\begin{equation}
\bigotimes\limits_{i=1}^{n}\left\vert \psi^{-}\right\rangle _{A_{i}B_{i}},
\end{equation}
each in the state $\left\vert \psi^{-}\right\rangle \equiv(|01\rangle
-|10\rangle)/\sqrt{2}$. One such scheme is as follows:

\begin{enumerate}
\item (Twirling) Bob chooses $n$ $2\times2$ unitaries $\{U_{1},\ldots,U_{n}\}$
at random and applies unitary $U_{i}$ to his $i$th qubit. He reports to Alice
which unitaries he selected and she applies $U_{i}$ to her $i$th qubit. This
\textquotedblleft twirling\textquotedblright\ step has the effect of
symmetrizing their state so that it is a mixture of Bell states.

\item For $i\in\{1,\ldots,n\}$, Bob picks one of the following three Pauli
operators%
\begin{equation}
X\equiv%
\begin{bmatrix}
0 & 1\\
1 & 0
\end{bmatrix}
,\ \ \ \ Z\equiv%
\begin{bmatrix}
1 & 0\\
0 & -1
\end{bmatrix}
,\ \ \ \ Y=iXZ
\end{equation}
at random. Let $P_{i}$ denote the $i$th choice. He measures $P_{i}$\ on his
$i$th qubit. After performing the last measurement, he sends all measurement
choices and outcomes to Alice.

\item For $i\in\{1,\ldots,n\}$, Alice measures $P_{i}$ on her $i$th qubit.

\item She accepts that the state is $n$ singlets if and only if all
measurement outcomes are different.
\end{enumerate}

The main reason that this 1-LOCC\ distinguishing protocol works is as follows:
the singlet is the only state having the property that measurement outcomes
are different when performing the same von Neumann measurement on each qubit
(for any von Neumann measurement). Furthermore, the
maximum probability with which a separable state can pass this test is equal
to 2/3, so that performing $n$ of these tests on a separable state
$\sigma_{A^{n}:B^{n}}$ reduces the probability of passing the
\textquotedblleft singlet test\textquotedblright\ to $\left(  2/3\right)
^{n}$.

We now analyze this protocol in more detail. Due to the fact that $\left(
U\otimes U\right)  \left\vert \psi^{-}\right\rangle =\left\vert \psi
^{-}\right\rangle $ for any $2\times2$ unitary $U$, the first step has no
effect on the singlets. Furthermore, the rest of the protocol succeeds with
probability one if the state is equal to $n$ singlets, due to the
property mentioned in the previous paragraph. So we turn to analyzing the probability of accepting
if the state is in fact separable. We begin by analyzing the first
\textquotedblleft Pauli test\textquotedblright\ in steps 2-4 and find a bound
on its acceptance probability. When doing so, it suffices to consider the
reduced state of $\sigma_{A^{n}:B^{n}}$ on systems $A_{1}$ and $B_{1}$, which
is separable across this cut because the original state is separable across
the $A^{n}:B^{n}$ cut. The initial twirling procedure transforms this
separable state to the following \textquotedblleft Werner\textquotedblright%
\ state:%
\begin{equation}
p\left\vert \psi^{-}\right\rangle \left\langle \psi^{-}\right\vert
_{A_{1}B_{1}}+\frac{1-p}{3}\left(  \left\vert \psi^{+}\right\rangle
\left\langle \psi^{+}\right\vert _{A_{1}B_{1}}+\left\vert \phi^{+}%
\right\rangle \left\langle \phi^{+}\right\vert _{A_{1}B_{1}}+\left\vert
\phi^{-}\right\rangle \left\langle \phi^{-}\right\vert _{A_{1}B_{1}}\right)  ,
\end{equation}
such that the maximal value of $p$ is $1/2$ \cite[Section~VI-B-9]{HHHH09}. (The states $\left\vert
\psi^{+}\right\rangle _{A_{1}B_{1}}$, $\left\vert \phi^{+}\right\rangle
_{A_{1}B_{1}}$, and $\left\vert \phi^{-}\right\rangle _{A_{1}B_{1}}$
are the other Bell states orthogonal to
$\left\vert \psi^{-}\right\rangle _{A_{1}B_{1}}$.) One can check that the
probability with which each of the three other Bell states
besides $\left\vert \psi^{-}\right\rangle _{A_{1}B_{1}}$ passes the \textquotedblleft
Pauli test\textquotedblright\ on the $i$th qubit (in steps 2-4 above) is equal
to 1/3. So this implies that the maximum probability with which this Pauli
test can pass is $1/2\cdot1+1/6\cdot(1/3+1/3+1/3)=2/3$. The analysis is the
same for the other $n-1$ Pauli tests:\ the only property that we use is that
the reduced states on systems $A_{i}:B_{i}$ is separable across this cut, so
that entanglement (or any correlation whatsoever) in the systems $A_{1}\cdots
A_{n}$ or $B_{1}\cdots B_{n}$ (but not across the cut $A^{n}:B^{n}$) cannot
help in passing this test. The result is that $\left(  2/3\right)  ^{n}$ is a
universal bound on the maximum probability with which any separable state
$\sigma_{A^{n}:B^{n}}$ can pass the overall test. By the discussion in
\expref{Sections}{sec:trace-dist} and \ref{sec:1-locc-dist}, the statement in the theorem follows.
\end{proof}

\section{\textsc{Pure Product State} is $\cls{BQP}$-complete}
\label{sec:bqp-complete}

We begin with the simplest of our separability testing promise problems---that of determining whether the state prepared by a given quantum circuit is close to a pure product state.
We propose two variants of this problem, one easier than the other.
We prove that the the harder variant is in $\cls{BQP}$ and we prove that the easier variant is $\cls{BQP}$-hard, establishing $\cls{BQP}$-completeness for both problems.

\begin{problem}
  [$(\alpha,\beta,l)$-\textsc{Pure Product State}]\footnote{If $l=2$ then the problem is called $(\alpha,\beta)$-\textsc{Bipartite Pure Product State}.  This convention applies to other problem names throughout the paper.}
  \label{problem:pure-product-state}
  \ \\[1mm]
  \begin{tabularx}{\textwidth}{lX}
    \emph{Input:} &
    A description of a quantum circuit that prepares an $l$-partite pure state $\ket{\psi}$.
    \\[1mm]
    \emph{Yes:} &
    $\ket{\psi}$ is $\alpha$-close to a pure product state:
      \begin{equation} \min_{\ket{\phi_1},\dots,\ket{\phi_l}} \Tnorm{ \psi - \phi_1\ot\cdots\ot\phi_l } \leq \alpha.   \end{equation}
    \\
    \emph{No:} &
    $\ket{\psi}$ is $\beta$-far from any pure product state:
      \begin{equation} \min_{\ket{\phi_1},\dots,\ket{\phi_l}} \Tnorm{ \psi - \phi_1\ot\cdots\ot\phi_l } \geq \beta.   \end{equation}
  \vspace{-\baselineskip}
  \end{tabularx}
\end{problem}

We define the \emph{one-way LOCC version} of $(\alpha,\beta,l)$-\textsc{Pure Product State} similarly except that the trace norm in the specification of a no-instance is replaced with the one-way LOCC norm.
The one-way LOCC version of \textsc{Pure Product State} trivially reduces to the trace distance version by virtue of the inequality $\tnorm{X}\geq\lnorm{X}$.
The main result of this section is the following theorem:

\begin{theorem}[\textsc{Pure Product State} is $\cls{BQP}$-complete]
\label{thm:bqp-completeness}

  The following hold:
  \begin{enumerate}
  \item The trace distance version of $(\alpha,\beta,l)$-\textsc{Pure Product State} is in $\cls{BQP}$ for all $l$ and all $\alpha<\beta\frac{\sqrt{11}}{32}$. (It is implicit here and throughout the rest of the paper that the gap between
  $\alpha$ and $\beta \frac{\sqrt{11}}{32}$ is larger than an inverse polynomial in the input length.)
  \item The one-way LOCC version of $(\varepsilon,2-\varepsilon)$-\textsc{Bipartite Pure Product State} is $\cls{BQP}$-hard, even when $\varepsilon$ decays exponentially in the input length.
  \end{enumerate}
  Thus, both problems are $\cls{BQP}$-complete for all $l\geq 2$ and all $(\alpha,\beta)$ with $0<\alpha<\beta\frac{\sqrt{11}}{32}$ and $\beta<2$.

\end{theorem}

\subsection{Membership in $\cls{BQP}$}

Our efficient quantum algorithm for the \textsc{Pure Product State} problem employs the \emph{product test}.
The product test is a boolean test that takes as input two copies of an arbitrary multipartite pure state $\ket{\psi}$.
The closer $\ket{\psi}$ is to a product state, the higher the probability with which the product test passes.
A specification of the product test is as follows:
\begin{enumerate}

\item Given are two copies of an arbitrary $l$-partite pure state $\ket{\psi}$.
  One of these copies is contained in registers $A_1,\dots,A_l$ and the other in $B_1,\dots,B_l$.

\item Perform $l$ swap tests---one for each pair of registers $(A_i,B_i)$ for $i=1,\dots,l$.
  Accept if and only if all the swap tests pass.

\end{enumerate}
The relationship between the distance from $\ket{\psi}$ to the nearest product state and the success probability of the product test was established in \cite{HM10}.

\begin{theorem}[\cite{HM10}] \label{thm:product-test}
  For each $l$-partite pure state $\ket{\psi}$ let $P_\textnormal{test}(\psi)$ denote the probability with which the product test passes when applied to $\ket{\psi}$ and let
    \begin{equation} 1-\varepsilon = \max_{\ket{\phi_1},\dots,\ket{\phi_l}} \abs{\braket{\psi}{\phi_1\ot\cdots\ot\phi_l}}^2. \end{equation}
  It holds that
    \begin{equation} 1-2\varepsilon \leq P_\textnormal{test}(\psi) \leq 1 - \frac{11}{512}\varepsilon. \end{equation}
\end{theorem}

The bounds of \expref{Theorem}{thm:product-test} are easily written in terms of the trace distance $t$ between $\ket{\psi}$ and the nearest product state via \eqref{eq:inner-product-to-tnorm}:
\begin{equation} \label{eq:product-test}
  1-t^2/2 \leq P_\textnormal{test}(\psi) \leq 1 - \frac{11}{2048}t^2.
\end{equation}
Armed with the product test, we now present our quantum algorithm for the \textsc{Pure Product State} problem.

\begin{proposition}
  $(\alpha,\beta,l)$-\textsc{Pure Product State} is in $\cls{BQP}$ for all $l$ and all $\alpha < \beta \frac{\sqrt{11}}{32}$.
\end{proposition}

\begin{proof}
The efficient quantum algorithm for $(\alpha,\beta,l)$-\textsc{Pure Product State} is as follows: use the input circuit to prepare two copies of $\ket{\psi}$, perform the product test, and accept if and only if the product test passes.

If $\ket{\psi}$ is a yes-instance then \eqref{eq:product-test} tells us that the product test passes with probability at least $1-\alpha^2/2$.
On the other hand, if $\ket{\psi}$ is a no-instance then \eqref{eq:product-test} tells us that the product test passes with probability at most $1-\frac{11}{2048}\beta^2$.
The algorithm witnesses membership in $\cls{BQP}$ whenever the former quantity is larger than the latter, which occurs whenever $\alpha<\beta\frac{\sqrt{11}}{32}$.
\end{proof}

\subsection{Hardness for $\cls{BQP}$}

\begin{proposition} \label{prop:bqp-hard}
  The one-way LOCC version of $(\varepsilon,2-\varepsilon)$-\textsc{Bipartite Pure Product State} is $\cls{BQP}$-hard, even when $\varepsilon$ decays exponentially in the input length.
\end{proposition}

\begin{proof}
  Let $L$ be any promise problem in $\cls{BQP}$ and let $\set{\ket{\nu}_x}_x$ be a family of efficiently preparable pure states witnessing membership of $L$ in $\cls{BQP}$.
  By this we mean the following: for each instance $x$ of $L$ the state $\ket{\nu_x}$ is held in two registers $DG$.
  Register $D$ is a decision qubit indicating acceptance or rejection of~$x$ and register $G$ is a garbage register that is a purifying system for $D$.

  Suppose that the family $\set{\ket{\nu}_x}_x$ has completeness $1-\delta$ and soundness $\delta$.
  In this proof we reduce the arbitrary problem $L$ to the one-way LOCC version of $(\alpha,\beta)$-\textsc{Bipartite Pure Product State} where
  \begin{align}
    \alpha &= 2\sqrt{\delta} \, ,\\
    \beta &= 2-2^{2-n/2}-2\sqrt{\delta} \, ,
  \end{align}
  for any desired $n$.
  The desired hardness result then follows by an appropriate choice of $\delta,n$, given that
  $\cls{BQP}(c,s) \subseteq  \cls{BQP}(\delta, 1-\delta)$ for any $\delta$ exponentially small in the input length.

  The reduction is as follows.
  Given an instance $x$ of $L$ we produce a description of the following circuit for preparing a pure state $\ket{\psi}$ of registers $AA'BDG$:
  \begin{enumerate}
  \item Prepare registers $AA'$ in a $2n$-qubit maximally entangled state such as $n$ EPR pairs, which we denote by $\ket{\phi^+}$.
    Prepare register $B$ in the $n$-qubit $\ket{0}$ state.
    Prepare registers $DG$ in state $\ket{\nu_x}$.
  \item Perform a controlled swap gate that swaps registers $A'$ and $B$ when $D$ is in the reject state $\ket{\textnormal{no}}$ and acts as the identity otherwise.
  \end{enumerate}
  A graphical depiction of this state preparation circuit appears later in the paper as a special case of \expref{Figure}{fig:qmareduction}.

  If $x$ is a yes-instance of $L$ then $\ket{\nu_x}$ has squared overlap at least $1-\delta$ with $\ket{\textnormal{yes}}_D\ket{\zeta}_G$ for some state $\ket{\zeta}$ of register $G$.
  It follows that the constructed state $\ket{\psi}$ is $2\sqrt{\delta}$-close in trace distance to $\ket{\phi^+}_{AA'}\ket{0}_B\ket{\textnormal{yes}}_D\ket{\zeta}_G$, which is product with respect to the cut $AA':BDG$. So $\ket{\psi}$ is a yes-instance of the one-way LOCC version of $(\alpha,\beta)$-\textsc{Bipartite Pure Product State}.

  Next, suppose that $x$ is a no-instance of $L$.
  In this case $\ket{\nu_x}$ has squared overlap at least $1-\delta$ with $\ket{\textnormal{no}}_D\ket{\eta}_G$ for some state $\ket{\eta}$ of register $G$.
  It follows that $\ket{\psi}$ is $2\sqrt{\delta}$-close in trace distance to a state which is in tensor product with the $2n$-qubit maximally entangled state $\ket{\phi^+}$ on registers $AB$.
  By contrast, for any product state $\ket{\phi}$ of registers $AA':BDG$ the reduced state $\ptr{A'DG}{\phi}$ of registers $AB$ must also be a product state.
  Thus, it suffices to exhibit a fixed one-way LOCC measurement that successfully distinguishes any product state of registers $AB$ from $n$ EPR pairs with high probability.
  The existence of such a measurement was proved in \expref{Theorem}{prop:owLOCC-to-sep}.

  We therefore have the following for any product state $\ket{\phi}$ of registers $AA':BDG$:
  \begin{align}
    \Lnorm{\psi - \phi} &\geq \Lnorm{\ptr{A'DG}{\phi} - \phi^+_{AB}} - \Lnorm{\phi^+_{AB} - \Ptr{A'DG}{\psi}} \\
    &\geq 2-2^{2-n/2} - 2\sqrt{\delta} \, ,
  \end{align}
  from which it follows that $\ket{\psi}$ is a no-instance of the one-way LOCC version of $(\alpha,\beta)$-\textsc{Bipartite Pure Product State}.
\end{proof}

\section{\textsc{Separable Isometry Output} (one-way LOCC version) is $\cls{QMA}$-complete}
\label{sec:qma-complete}

In this section we prove $\cls{QMA}$-completeness of the problem of deciding whether the isometry implemented by a given quantum circuit can be made to produce a state that is close to separable in trace distance or far from separable in one-way LOCC distance.

\begin{problem}
  [$(\alpha,\beta,l)$-\textsc{Separable Isometry Output}, one-way LOCC version]
  \label{problem:separable-isometry-output-1-locc}
  \ \\[1mm]
  \begin{tabularx}{\textwidth}{lX}
    \emph{Input:} &
    A description of a quantum circuit that implements an isometry $U$ with an $l$-partite output system $A_1\cdots A_l$.
    \\[1mm]
    \emph{Yes:} &
    There is an input state $\rho$ such that $U\rho U^*$ is $\alpha$-close in trace distance to separable:
    \begin{equation}  \min_\rho \min_{\sigma\in\mathcal{S}(A_1:\cdots:A_l)} \Tnorm{ U\rho U^* - \sigma } \leq \alpha. \label{eq:yes-qma-prom} \end{equation}
    \\
    \emph{No:} &
    For all input states $\rho$ it holds that $U\rho U^*$ is $\beta$-far in one-way LOCC distance from separable:
    \begin{equation} \min_\rho \min_{\sigma\in\mathcal{S}(A_1:\cdots:A_l)} \Lnorm{ U\rho U^* - \sigma } \geq \beta. \end{equation}
  \vspace{-\baselineskip}
  \end{tabularx}
\end{problem}

The main result of this section is the following theorem:

\begin{theorem}[\textsc{Separable Isometry Output}, one-way LOCC version is $\cls{QMA}$-complete]
\label{thm:qma-completeness}

  The following hold:
  \begin{enumerate}
  \item The one-way LOCC version of $(\alpha,\beta,l)$-\textsc{Separable Isometry Output} is in $\cls{QMA}$ for all $l$ and all $\alpha < \beta^4/16$.
  \item The one-way LOCC version of $(\varepsilon,2-\varepsilon)$-\textsc{Bipartite Separable Isometry Output} is $\cls{QMA}$-hard, even when $\varepsilon$ decays exponentially in the input length.
  \end{enumerate}
  Thus, the problem is $\cls{QMA}$-complete for all $l\geq 2$, all $0<\alpha<\beta^4/16$, and all $\beta<2$.

\end{theorem}

\subsection{Containment in $\cls{QMA}$}

Our quantum witness for separability invokes the notion of $k$-extendibility of separable states \cite{W89a}.
We therefore begin with a brief summary of $k$-extendibility.

Let $AB$ be any two registers and let $B_1,\dots,B_k$ be registers each of the same size as $B$.
A bipartite state $\rho$ of registers $AB$ is \emph{$k$-extendible} if there exists a state $\omega$ of registers $AB_1\cdots B_k$
that is invariant under permutations of registers $B_1,\dots,B_k$ and
 consistent with $\rho$, meaning that $\ptr{B_2\cdots B_k}{\omega}=\rho$.

The set of all $k$-extendible states (with respect to a given cut of the registers) is denoted $\mathcal{E}_k$.
It is a basic fact that every separable state is $k$-extendible for all $k$, so that
\( \mathcal{S}\subseteq\mathcal{E}_k. \)
To see this, let
\begin{equation} \rho = \sum_i p_i \ketbraa{{\psi^i}}_{A} \ot \ketbraa{{\phi^i}}_{B} \end{equation}
be any separable state of registers $AB$ and observe that
\begin{equation} \sum_i p_i \ketbraa{{\psi^i}}_{A}\ot\ketbraa{{\phi^i}}_{B}^{\ot k}
\label{eq:pure-state-k-ext}
\end{equation}
is a $k$-extension of $\rho$.
It is known that if $\rho$ is not separable then there exists some $k'$ for which $\rho$ is not $k'$-extendible.
Moreover, it is known that $\mathcal{E}_{k+1}\subseteq\mathcal{E}_k$ for all $k$, from which it follows that the sets $\mathcal{E}_k$ form a containment hierarchy that converges to the set $\mathcal{S}$ of separable states in the limit $k\to\infty$ \cite{DPS02,DPS04}.

The notion of $k$-extendibility extends naturally to multi-register systems $A_1\cdots A_l$ by imposing the extendibility condition on each individual register \cite{DPS05,BH12}, though the notation is cumbersome.
Formally, let $A_{i,1},\dots,A_{i,k}$ be registers of the same size as $A_i$.
A state $\rho$ of registers $A_1\cdots A_l$ is $k$-extendible with respect to $A_1:\cdots:A_l$ if there exists a global state $\omega$ of all $lk$ registers $A_{i,j}$ that is consistent with $\rho$ on $A_1\cdots A_l$ and invariant under permutations of registers $A_{i,1},\dots,A_{i,k}$ for all $i=1,\dots,l$.
(Observe that there are $l\cdot k!$ such permutations.)

Brand\~ao and Harrow have shown that if $\rho$ is close to $k$-extendible in trace distance for not-too-large $k$ then $\rho$ is also close to separable in one-way LOCC distance \cite{BH12}.
The following is a straightforward consequence of their result.

\begin{theorem}
\label{lem:multi-k-ext}

Let $A_1,\dots,A_l$ be registers whose total combined dimension is $D$.
Let $\rho$ be $\varepsilon$-far from separable in one-way LOCC distance, so that
\begin{equation} \min_{\sigma \in \mathcal{S}(A_1:\cdots:A_l)} \lnorm{\rho-\sigma} \geq \varepsilon. \end{equation}
Then for any $\delta<\varepsilon$ it holds that $\rho$ is $\delta$-far from $k$-extendible in trace distance, so that
\begin{equation} \min_{\sigma' \in \mathcal{E}_k(A_1:\cdots:A_l)} \tnorm{\rho-\sigma'} \geq \delta, \end{equation}
provided
\begin{equation} k \geq \left\lceil l + \frac{4l^2\log D}{(\varepsilon - \delta)^2}\right\rceil. \end{equation}

\end{theorem}

\begin{proof}
Let $\sigma'$ be any $k$-extendible state with $k>l$.
We know from \cite[Theorem~2 and Corollary~8]{BH12} that there is a separable state $\sigma''\in\mathcal{S}$ such that
\begin{equation} \label{eq:1locc-multi-sigma-min}
  \Lnorm{\sigma'-\sigma''} \leq \sqrt{\frac{4 l^2 \log D }{k - l}}.
\end{equation}
So we use this in the following chain of inequalities:
\begin{align}
  \varepsilon
  &\leq \min_{\sigma \in \mathcal{S}(A_1:\cdots:A_l)} \lnorm{\rho-\sigma}  \\
  &\leq \lnorm{\rho-\sigma''} \\
  & \leq \lnorm{\rho-\sigma'} + \lnorm{\sigma'-\sigma''} \\
  & \leq \tnorm{\rho-\sigma'} + \sqrt{\frac{4 l^2 \log D }{k - l}}
  \end{align}
Since this bound holds for any $k$-extendible state, we can conclude that
\begin{equation}
\varepsilon - \sqrt{\frac{4 l^2 \log D }{k - l}}
\leq \min_{\sigma' \in \mathcal{E}_k(A_1:\cdots:A_l)} \tnorm{\rho-\sigma'} \, .
\end{equation}
The statement of the theorem then follows by picking $k$ large enough so that
$\varepsilon - \sqrt{\frac{4 l^2 \log D }{k - l}} \geq \delta$.
\end{proof}

We now present our succinct quantum witness for the one-way LOCC version of the \textsc{Separable Isometry Output} problem.

\begin{proposition}
  The one-way LOCC version of $(\alpha,\beta,l)$-\textsc{Separable Isometry Output} is in $\cls{QMA}$ for all $l$ and all $\alpha<\beta^4/16$.
\end{proposition}

\begin{proof}
  For convenience we write $A\equiv A_1 \cdots A_l$ where the combined register $A$ has dimension $D$.
  It is helpful to label the input and output registers of $U$ as $U:S\to A$.
  Let $\varepsilon>0$ be such that $\sqrt{\alpha}<(\beta-\varepsilon)^2/4$.
  The verifier witnessing membership of the problem in $\cls{QMA}$ is as follows:
  \begin{enumerate}

  \item
    Receive $kl+1$ registers from the prover labeled $S$ and $A_i^j$ where $A_i^j$ has the same size as $A_i$ for $i=1,\dots,l$ and $j=1,\dots,k$ and
    \begin{equation} k = \left\lceil l + \frac{4l^2\log D}{\varepsilon^2} \right\rceil. \end{equation}
    Apply $U$ to register $S$ to obtain register $A\equiv A_1,\dots,A_l$.

  \item \label{it:perm-tests}
    Perform $l$ permutation tests: one for each group $(A_i,A_i^1,\dots,A_i^k)$ of $k+1$ registers.
    Accept if and only if all permutation tests pass.

  \end{enumerate}
  In what follows we use the shorthand $A^j\equiv A_1^j\cdots A_l^j$ for each $j\in\set{1,\dots,k}$.

  Suppose first that $U$ is a yes-instance of the problem.
  We show that there exists a state $\rho_{SA^1 \cdots A^k}$ of the $kl+1$ registers $SA^1 \cdots A^k$ that causes the verifier to accept with probability at least $1-\sqrt{\alpha}$.
  To this end we define the following symbols:
  \begin{enumerate}
  \item Let $\sigma_{A}$ be a separable state and $\rho_S$ be a state of register $S$ such that \begin{equation}\tnorm{U\rho_S U^*-\sigma_{A}}\leq\alpha \, , \label{eq:prom-for-states}\end{equation}
  as promised in \eqref{eq:yes-qma-prom}.
  \item Let $\sigma_{AA^1 \cdots A^k}$ be a $(k+1)$-extension of $\sigma_{A}$ in registers $AA^1 \cdots A^k$. It is important that this $(k+1)$-extension be taken as a convex combination of pure states as in \eqref{eq:pure-state-k-ext}, so that it would be accepted by the permutation test with probability one.
  \item Let $\hat{U}:SW\to A$ denote the unitary circuit that implements $U$ when the workspace register $W$ is initialized to $\ket{0}$.
  \end{enumerate}
  By the preservation of subsystem fidelity \cite[Lemma 7.2]{JUW09} there exists a state $\rho_{SA^1 \cdots A^k}$ of registers $SA^1 \cdots A^k$ consistent with $\rho_S$ such that
  \begin{equation} \label{eq:subsystem-fid}
    F\Pa{\pa{\hat{U}^*\ot I_{A^1 \cdots A^k}}\sigma_{AA^1 \cdots A^k}\pa{\hat{U}\ot I_{A^1 \cdots A^k}}, \rho_{SA^1 \cdots A^k}\ot\kb{0}_W} = F\Pa{\hat{U}^*\sigma_{A}\hat{U}, \rho_S\ot\kb{0}_W}.
  \end{equation}
  Let us argue that this state $\rho_{SA^1 \cdots A^k}$ is our desired state.
   It follows from \eqref{eq:FvG-ineqs}, \eqref{eq:prom-for-states}, and unitary
   invariance of fidelity that
  \begin{align}
  1-\alpha & \leq F\Pa{\sigma_{A}, \hat{U} (\rho_S\ot\kb{0}_W) \hat{U}^*}\\
  &   = F\Pa{\hat{U}^*\sigma_{A}\hat{U}, \rho_S\ot\kb{0}_W} \, .
  \end{align}
  Applying the above and \eqref{eq:FvG-ineqs} to the right side of \eqref{eq:subsystem-fid}, we find that the quantity in \eqref{eq:subsystem-fid} is at least
  \begin{equation} \label{eq:rhs-fid}
    1- \Tnorm{\hat{U}^*\sigma_{A}\hat{U} - \rho_S\ot\kb{0}_W} \geq 1 - \alpha.
  \end{equation}
  Applying \eqref{eq:FvG-ineqs} to the left side of \eqref{eq:subsystem-fid}, we find that the quantity in \eqref{eq:subsystem-fid} is at most
  \begin{equation} \label{eq:lhs-fid}
    1 - \frac{1}{4} \Tnorm{ \pa{\hat{U}^* \ot I_{A^1 \cdots A^k}}\sigma_{AA^1 \cdots A^k}\pa{\hat{U}\ot I_{A^1 \cdots A^k}} - \rho_{SA^1 \cdots A^k}\ot\kb{0}_W }^2.
  \end{equation}
  Combining \eqref{eq:rhs-fid} and \eqref{eq:lhs-fid} leads to the following bound:
  \begin{equation} \Tnorm{ \sigma_{AA^1\cdots A^k} - \pa{U\ot I_{A^1\cdots A^k}}\rho_{SA^1\cdots A^k}\pa{ U^* \ot I_{A^1\cdots A^k}} } \leq 2\sqrt{\alpha}. \end{equation}
  Thus, $\rho_{SA^1\cdots A^k}$ is $2\sqrt{\alpha}$-close in trace distance to a state that is accepted by the verifier with certainty.
  It then follows from \eqref{eq:trace-inequality} that the verifier accepts $\rho_{SA^1\cdots A^k}$ with probability at least $1-\sqrt{\alpha}$ as desired.

  Now suppose that $U$ is a no-instance of the problem.
  By \expref{Theorem}{lem:multi-k-ext} and our choice of $k$ we have that
  \begin{equation} \min_{\rho_S} \min_{\sigma_A\in\mathcal{E}_k(A_1:\cdots:A_l)} \Tnorm{U\rho_S U^* - \sigma_A} \geq \beta-\varepsilon. \end{equation}
  We claim that an upper bound on the probability with which all the permutation tests pass is given by the maximum fidelity of $U\rho_S U^*$ with a $k$-extendible state:
  \begin{equation} \Pr[\textnormal{all pass}] \leq \max_{\rho_S} \max_{\sigma_A\in\mathcal{E}_k} F\Pa{U\rho_S U^*,\sigma_A}. \label{eq:upp-bound-acc} \end{equation}
  It follows from \eqref{eq:FvG-ineqs} that this probability is at most $1-(\beta-\varepsilon)^2/4$.
  We chose $\varepsilon$ so that the completeness $1-\sqrt{\alpha}$ is larger than the soundness $1-(\beta-\varepsilon)^2/4$, from which it follows that the problem is in $\cls{QMA}$.

  We now justify the claim in \eqref{eq:upp-bound-acc} using a method similar to that in \cite[Section~4]{HMW14}.
  In order to implement the permutation tests in step \ref{it:perm-tests} the verifier prepares a control register $C$ in state
  \begin{equation}
    \left\vert \text{perm}\right\rangle _{C}\equiv\frac{1}{\sqrt{k!}}\sum_{\pi\in S_{k}}\left\vert \pi\right\rangle _{C},\label{eq:perm-state}%
  \end{equation}
  which is a uniform superposition over all possible permutations of $k$ elements resulting from an application of the quantum Fourier transform \cite{NC00} to the state $\left\vert 0\right\rangle _{C}$, so that the $C$ register requires $\left\lceil \log_{2}\left(  k!\right)  \right\rceil $ qubits.
  The verifier then applies the following controlled-permutation operation:%
\begin{equation}
\left(  U_{\Pi}\right)  _{AA^{1}\cdots A^{k}C}\equiv\sum_{\pi\in S_{k}%
}W_{AA^{1}\cdots A^{k}}^{\pi}\otimes\left\vert \pi\right\rangle \left\langle
\pi\right\vert _{C},\label{eq:controlled-perm}%
\end{equation}
where $W_{AA^{1}\cdots A^{k}}^{\pi}$ is a unitary operation corresponding to
permutation $\pi$. The verifier finally applies an inverse quantum Fourier
transform to $C$, measures it in the computational basis, and accepts if the
measurement outcomes are all zeros. Letting $\left\vert \psi\right\rangle
_{RSA^{1}\cdots A^{k}}$ be a purification of the prover's input, we can write
the maximum acceptance probability of this proof system as follows:%
\begin{multline}
\max_{\left\vert \psi\right\rangle _{RSA^{1}\cdots A^{k}}}\left\Vert
\left\langle 0\right\vert _{C}\text{QFT}_{C}^{-1}\left(  U_{\Pi}\right)
_{AA^{1}\cdots A^{k}C}U_{S\rightarrow A}\left\vert \psi\right\rangle
_{RSA^{1}\cdots A^{k}}\left\vert \text{perm}\right\rangle _{C}\right\Vert
_{2}^{2}\label{eq:max-acc-qma}\\
=\max_{\left\vert \psi\right\rangle _{RSA^{1}\cdots A^{k}},\left\vert
\phi\right\rangle _{RAA^{1}\cdots A^{k}}}\left\vert \left\langle 0\right\vert
_{C}\left\langle \phi\right\vert _{RAA^{1}\cdots A^{k}}\text{QFT}_{C}%
^{-1}\left(  U_{\Pi}\right)  _{AA^{1}\cdots A^{k}C}U_{S\rightarrow
A}\left\vert \psi\right\rangle _{RSA^{1}\cdots A^{k}}\left\vert \text{perm}%
\right\rangle _{C}\right\vert _{2}^{2}%
\end{multline}
We can define a channel generated by the inverse of the verifier's circuit
conditioned on accepting as follows:%
\begin{equation}
\mathcal{M}_{AA^{1}\cdots A^{k}\rightarrow AC}\left(  \sigma_{AA^{1}\cdots
A^{k}}\right)  \equiv\text{Tr}_{A^{1}\cdots A^{k}}\left\{  \left(  U_{\Pi
}\right)  _{AA^{1}\cdots A^{k}C}\left(  \sigma_{AA^{1}\cdots A^{k}}%
\otimes\left\vert \text{perm}\right\rangle \left\langle \text{perm}\right\vert
_{C}\right)  \left(  U_{\Pi}^{\ast}\right)  _{AA^{1}\cdots A^{k}C}\right\}  .
\end{equation}
After doing so, we can apply Uhlmann's theorem to (\ref{eq:max-acc-qma}) to
rewrite the maximum acceptance probability as follows:%
\begin{equation}
\max_{\rho_{S},\sigma_{AA^{1}\cdots A^{k}}}F\left(  U_{S\rightarrow A}\rho
_{S}U_{S\rightarrow A}^{\ast}\otimes\left\vert \text{perm}\right\rangle
\left\langle \text{perm}\right\vert _{C},\mathcal{M}_{AA^{1}\cdots
A^{k}\rightarrow AC}\left(  \sigma_{AA^{1}\cdots A^{k}}\right)  \right)  .
\end{equation}
Since the fidelity can only increase under the discarding of the control
register $C$,\footnote{We can interpret discarding the control register as
actually giving it to the prover, so that the resulting fidelity corresponds
to the maximum acceptance probability in a modified protocol in which the
prover controls the inputs to $C$.} the maximum acceptance probability is
upper bounded by the following quantity:%
\begin{equation}
\max_{\rho_{S},\sigma_{AA^{1}\cdots A^{k}}}F\left(  U_{S\rightarrow A}\rho
_{S}U_{S\rightarrow A}^{\ast},\mathcal{M}_{AA^{1}\cdots A^{k}\rightarrow
A}\left(  \sigma_{AA^{1}\cdots A^{k}}\right)  \right)
,\label{eq:max-k-ext-fidelity-1}%
\end{equation}
where%
\begin{align}
\mathcal{M}_{AA^{1}\cdots A^{k}\rightarrow A}\left(  \sigma_{AA^{1}\cdots
A^{k}}\right)   &  =\text{Tr}_{C}\left\{  \mathcal{M}_{AA^{1}\cdots
A^{k}\rightarrow AC}\left(  \sigma_{AA^{1}\cdots A^{k}}\right)  \right\}  \\
&  =\frac{1}{k!}\sum_{\pi\in S_{k}}\text{Tr}_{A^{1}\cdots A^{k}}\left\{
W_{AA^{1}\cdots A^{k}}^{\pi}\sigma_{AA^{1}\cdots A^{k}}\left(  W_{AA^{1}\cdots
A^{k}}^{\pi}\right)  ^{\ast}\right\}  ,\nonumber
\end{align}
The equation above reveals that $\mathcal{M}_{AA^{1}\cdots A^{k}\rightarrow
A}$ is just the channel that applies a random permutation of the $AA^{1}\cdots
A^{k}$ systems and discards the last $k$ systems $A^{1}\cdots A^{k}$. Clearly,
since the channel $\mathcal{M}_{AA^{1}\cdots A^{k}\rightarrow A}$ symmetrizes
the state of the systems $AA^{1}\cdots A^{k}$, the maximum in
(\ref{eq:max-k-ext-fidelity-1}) is achieved by a state $\sigma_{AA^{1}\cdots
A^{k}}$ for which systems $AA^{1}\cdots A^{k}$ are permutation symmetric.
Thus, by recalling the definition of $k$-extendibility, we can rewrite
(\ref{eq:max-k-ext-fidelity-1}) as the maximum $k$-extendible fidelity of
$U_{S\rightarrow A}\rho_{S}U_{S\rightarrow A}^{\ast}$:%
\begin{equation}
\max_{\rho_{S},\sigma_{AA^{1}\cdots A^{k}}}F\left(  U_{S\rightarrow A}\rho
_{S}U_{S\rightarrow A}^{\ast},\mathcal{M}_{AA^{1}\cdots A^{k}\rightarrow
A}\left(  \sigma_{AA^{1}\cdots A^{k}}\right)  \right)  =\max_{\rho_{S}%
,\sigma_{A}\in\mathcal{E}_{k}\left(  A_{1}:\cdots:A_{l}\right)  }F\left(
U_{S\rightarrow A}\rho_{S}U_{S\rightarrow A}^{\ast},\sigma_{A}\right)
.\label{eq:max-exists-arg}%
\end{equation}
This demonstrates that the maximum $k$-extendible fidelity is an upper bound
on the maximum acceptance probability and completes our proof of the claim
in \eqref{eq:upp-bound-acc}.
\end{proof}

\subsection{Hardness for $\cls{QMA}$}

\begin{proposition} \label{prop:qma-hard}
  The one-way LOCC version of $(\varepsilon,2-\varepsilon)$-\textsc{Bipartite Separable Isometry Output} is $\cls{QMA}$-hard, even when $\varepsilon$ decays exponentially in the input length.
\end{proposition}

\begin{proof}
  This proof is almost exactly the same as the proof of \expref{Proposition}{prop:bqp-hard}.
  The only difference is that here we must quantify over all states of a new input register $P$ for each circuit.
  Nonetheless, we include a full proof for completeness.

  Let $L$ be any promise problem in $\cls{QMA}$ and let $\set{V_x}_x$ be a family of isometric verifier circuits
  witnessing this fact with completeness $1-\delta$ and soundness $\delta$ for sufficiently small $\delta$ to be chosen later.
  Circuits in this family take the form $V_x:P\to DG$.
  The input register $P$ is supplied by the prover.
  The output register $D$ is a decision qubit indicating acceptance or rejection of $x$ and the output register $G$ is a garbage register that holds the purification of $D$.

  In this proof we reduce the arbitrary problem $L$ to the one-way LOCC version of $(\alpha,\beta)$-\textsc{Bipartite Separable Isometry Output} where
  \begin{align}
    \alpha &= 2\sqrt{\delta} \, ,\\
    \beta &= 2-2^{2-n/2}-2\sqrt{\delta}\, ,
  \end{align}
  for any desired $n$.
  The desired hardness result then follows by an appropriate choice of $\delta,n$.

  The reduction is as follows.
  Given an instance $x$ of $L$ we produce a description of the following isometric circuit $U:P\to AA'BDG$:
  \begin{enumerate}
  \item Given the input register $P$ apply the verifier circuit $V_x$ to obtain registers $DG$.
  \item Prepare registers $AA'$ in a $2n$-qubit maximally entangled state such as $n$ EPR pairs, which we denote by $\ket{\phi^+}$.
    Prepare register $B$ in the $n$-qubit $\ket{0}$ state.
  \item Perform a unitary conditional swap gate that swaps registers $A'$ and $B$ when $D$ is in the reject state $\ket{\textnormal{no}}$ and acts as the identity otherwise.
  \end{enumerate}
  See \expref{Figure}{fig:qmareduction} for a graphical depiction of this circuit.
  \begin{figure}
    \begin{center}
    \includegraphics[scale=.6]{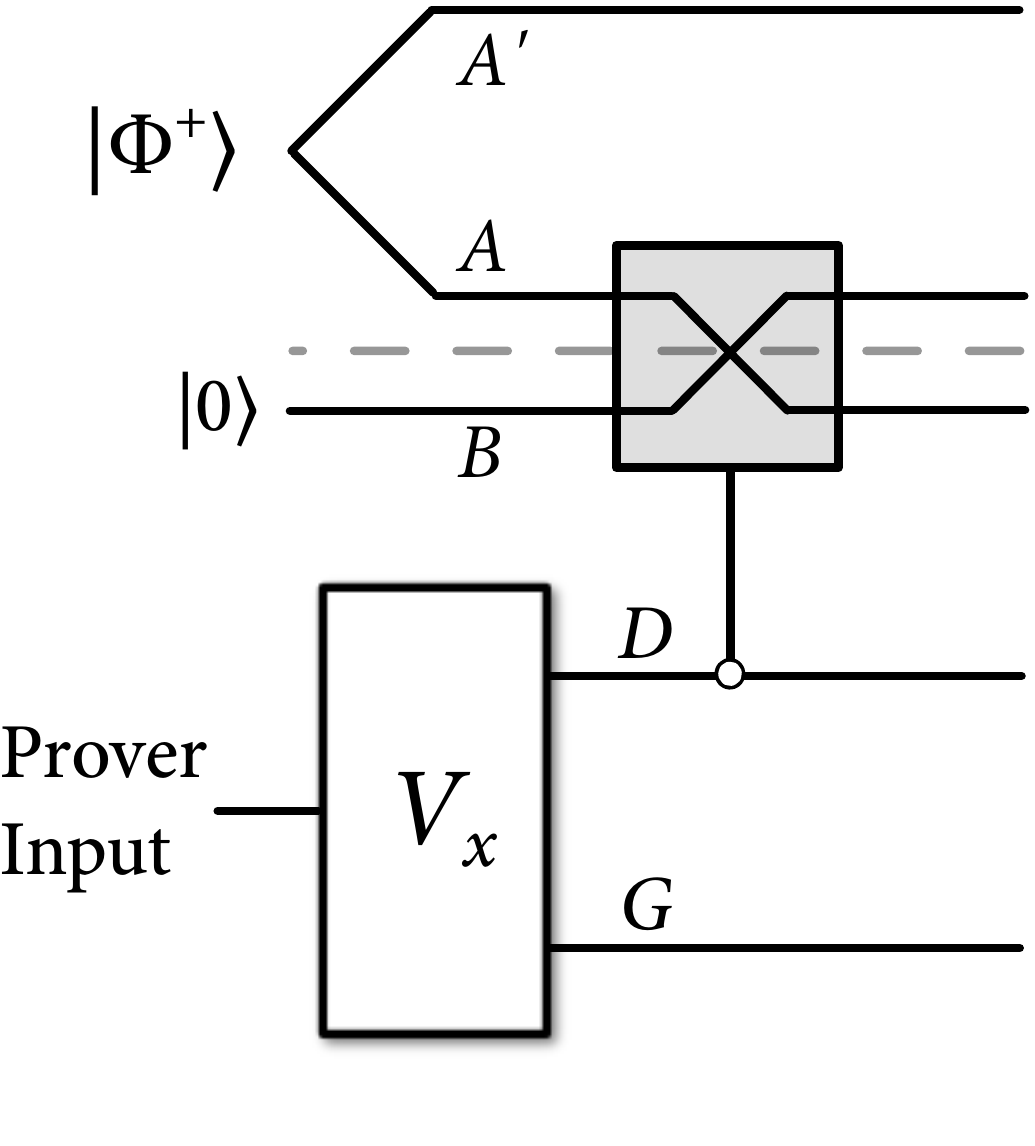}
    \caption{\label{fig:qmareduction}
      The circuit $U$ produced by our reduction from an arbitrary problem $L\in\cls{QMA}$ to the one-way LOCC version of $(\varepsilon,2-\varepsilon)$-\textsc{Bipartite Separable Isometry Output}.
      The dashed line indicates that the output registers are to be divided along the bipartite cut $AA':BDG$.
      This construction also appears in the proof of \expref{Proposition}{prop:bqp-hard} in the special case where the prover's input is empty.}
    \end{center}
  \end{figure}

  Suppose $x$ is a yes-instance of $L$ and let $\ket{\varpi}$ be a pure state of register $P$ that causes the verifier to accept with high probability, meaning that the state $V_x\ket{\varpi}$ has squared overlap at least $1-\delta$ with $\ket{\textnormal{yes}}_D\ket{\zeta}_G$ for some state $\ket{\zeta}$ of register $G$.
  It follows that $U\ket{\varpi}$ is $2\sqrt{\delta}$-close in trace distance to $\ket{\phi^+}_{AA'}\ket{0}_B\ket{\textnormal{yes}}_D\ket{\zeta}_G$, which is product with respect to the cut $AA':BDG$, and so $U$ is a yes-instance of the one-way LOCC version of $(\alpha,\beta)$-\textsc{Bipartite Separable Isometry Output}.

  Next, suppose that $x$ is a no-instance of $L$.
  In this case for all input states $\rho$ of register $P$ the output state $U\rho U^*$ of registers $AA'BDG$ is $2\sqrt{\delta}$-close to a state which is in tensor product with the $2n$-qubit maximally entangled state $\ket{\phi^+}$ on registers $AB$.
  By contrast, for any separable state $\sigma$ of registers $AA':BDG$ the reduced state $\ptr{A'DG}{\sigma}$ of registers $AB$ must also be separable.
  Thus, it suffices to exhibit a fixed one-way LOCC measurement that successfully distinguishes any separable state of registers $AB$ from $n$ EPR pairs with high probability.
  The existence of such a measurement was proved in \expref{Theorem}{prop:owLOCC-to-sep}.

  We therefore have the following for any input state $\rho$ of register $P$ and any separable state $\sigma$ of registers $AA':BDG$:
  \begin{align}
    \Lnorm{U\rho U^* - \sigma} &\geq \Lnorm{\ptr{A'DG}{\sigma} - \phi^+_{AB}} - \Lnorm{\phi^+_{AB} - \Ptr{A'DG}{U\rho U^*}} \\
    &\geq 2-2^{2-n/2} - 2\sqrt{\delta}\, ,
  \end{align}
  from which it follows that $U$ is a no-instance of the one-way LOCC version of $(\alpha,\beta)$-\textsc{Bipartite Separable Isometry Output}.
\end{proof}

\section{\textsc{Separable Isometry Output} is $\cls{QMA}(2)$-complete}
\label{sec:qma2-complete}

In \expref{Section}{sec:qma-complete} we showed that the one-way LOCC version of the \textsc{Separable Isometry Output} problem is $\cls{QMA}$-complete.
By contrast, in this section we show that the trace distance version of this problem (and some closely related variants of it) are $\cls{QMA}(2)$-complete.

We begin by restricting attention to the problem of determining whether an isometry $U$ described by a quantum circuit can be made to produce a pure product output state from a pure input state.

\begin{problem}
  [$(\alpha,\beta,l)$-\textsc{Pure Product Isometry Output}]
  \label{problem:product-isometry-output}
  \ \\[1mm]
  \begin{tabularx}{\textwidth}{lX}
    \emph{Input:} &
    A description of a quantum circuit that implements an isometry $U$ with an $l$-partite output system $A_1\cdots A_l$.
    \\[1mm]
    \emph{Yes:} &
    There is an input state $\ket{\psi}$ such that $U\ket{\psi}$ is $\alpha$-close to a pure product state:
    \begin{equation} \min_{\ket{\psi}} \min_{\ket{\phi_1},\dots,\ket{\phi_l}} \Tnorm{ U\psi U^* - \phi_1\ot\cdots\ot\phi_l } \leq \alpha. \end{equation}
    \\
    \emph{No:} &
    For all input states $\ket{\psi}$ it holds that $U\ket{\psi}$ is $\beta$-far from a pure product state:
    \begin{equation} \min_{\ket{\psi}} \min_{\ket{\phi_1},\dots,\ket{\phi_l}} \Tnorm{ U\psi U^* - \phi_1\ot\cdots\ot\phi_l } \geq \beta. \end{equation}
  \vspace{-\baselineskip}
  \end{tabularx}
\end{problem}

The main result of this section is the following theorem:

\begin{theorem}[\textsc{Pure Product Isometry Output} is $\cls{QMA}(2)$-complete]
\label{thm:qma2-complete}

  The following hold:
  \begin{enumerate}
  \item $(\alpha,\beta,l)$-\textsc{Pure Product Isometry Output} is in $\cls{QMA}(2)$ for all $l$ and all $\alpha < \beta$.
  \item $(\varepsilon,2-\varepsilon)$-\textsc{Bipartite Pure Product Isometry Output} is $\cls{QMA}(2)$-hard, even when $\varepsilon$ decays exponentially in the input length.
  \end{enumerate}
  Thus, the problem is $\cls{QMA}(2)$-complete for all $l\geq 2$ and all $0<\alpha<\beta<2$.

\end{theorem}

\subsection{Containment in $\cls{QMA}(2)$}

\begin{proposition}
  $(\alpha,\beta,l)$-\textsc{Pure Product Isometry Output} is in $\cls{QMA}(2)$ for all $l$ and all $\alpha < \beta$.
\end{proposition}

\begin{proof}
  We prove that the problem is in $\cls{QMA}(l+1)$, from which it follows that the problem is also in $\cls{QMA}(2)$ via the main result of \cite{HM10}.
  The verifier witnessing membership of the problem in $\cls{QMA}(l+1)$ is as follows:
  \begin{enumerate}
  \item Receive an input state $\ket{\psi}$ from one of the provers and a candidate product state $\ket{\phi_1}\ot\cdots\ot\ket{\phi_l}$ from the remaining $l$ provers.
  \item Apply $U$ to the input.
    Perform a swap test between $U\ket{\psi}$ and $\ket{\phi_1}\ot\cdots\ot\ket{\phi_l}$.
    Accept if and only if the swap test passes.
  \end{enumerate}
  If $U$ is a yes-instance then the provers can cause the verifier to accept with probability at least $1-\alpha^2/8$ by an appropriate choice of states $\ket{\psi},\ket{\phi_1},\dots,\ket{\phi_l}$.
      It follows from a standard convexity argument that the provers achieve their maximum probability of success for the swap test when they each send the verifier a pure state, so we assume that they do so without loss of generality.
  So if $U$ is a no-instance then the verifier will accept with probability at most $1-\beta^2/8$ regardless of which states the provers send to the verifier.
  As $\alpha<\beta$, there is a gap between completeness and soundness for this verifier.
\end{proof}

\subsection{Hardness for $\cls{QMA}(2)$}

\begin{proposition}
  $(\varepsilon,2-\varepsilon)$-\textsc{Bipartite Pure Product Isometry Output} is $\cls{QMA}(2)$-hard, even when $\varepsilon$ decays exponentially in the input length.
\end{proposition}

\begin{proof}
  Let $L$ be any promise problem in $\cls{QMA}(2)$ and let $\set{V_x}_x$ be a family of unitary verifier circuits (indexed by instances $x$ of $L$) witnessing this fact with completeness $1-\delta$ and soundness $\delta$ for sufficiently small $\delta$ to be chosen later.
  Circuits in this family take the form $V_x:ABW\to DG$.
  Such a verifier circuit $V_x$ is depicted in \expref{Figure}{fig:qma2-reduction}(a).
  \begin{figure}
    \begin{center}
    \begin{tabular}{b{6cm} b{7cm}}
      (a)\includegraphics[scale=.6]{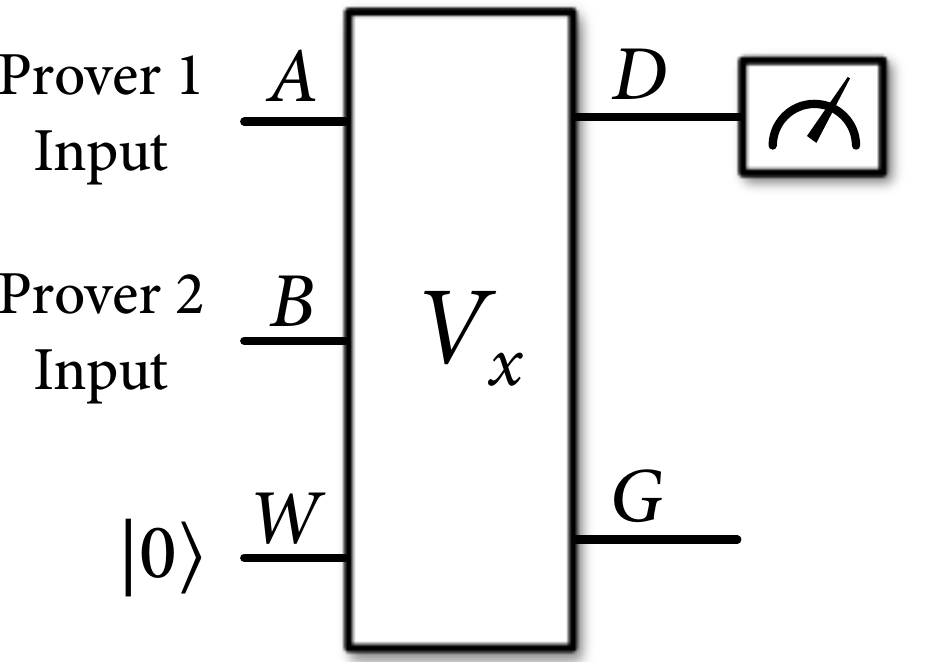} &
      (b)\includegraphics[scale=.5]{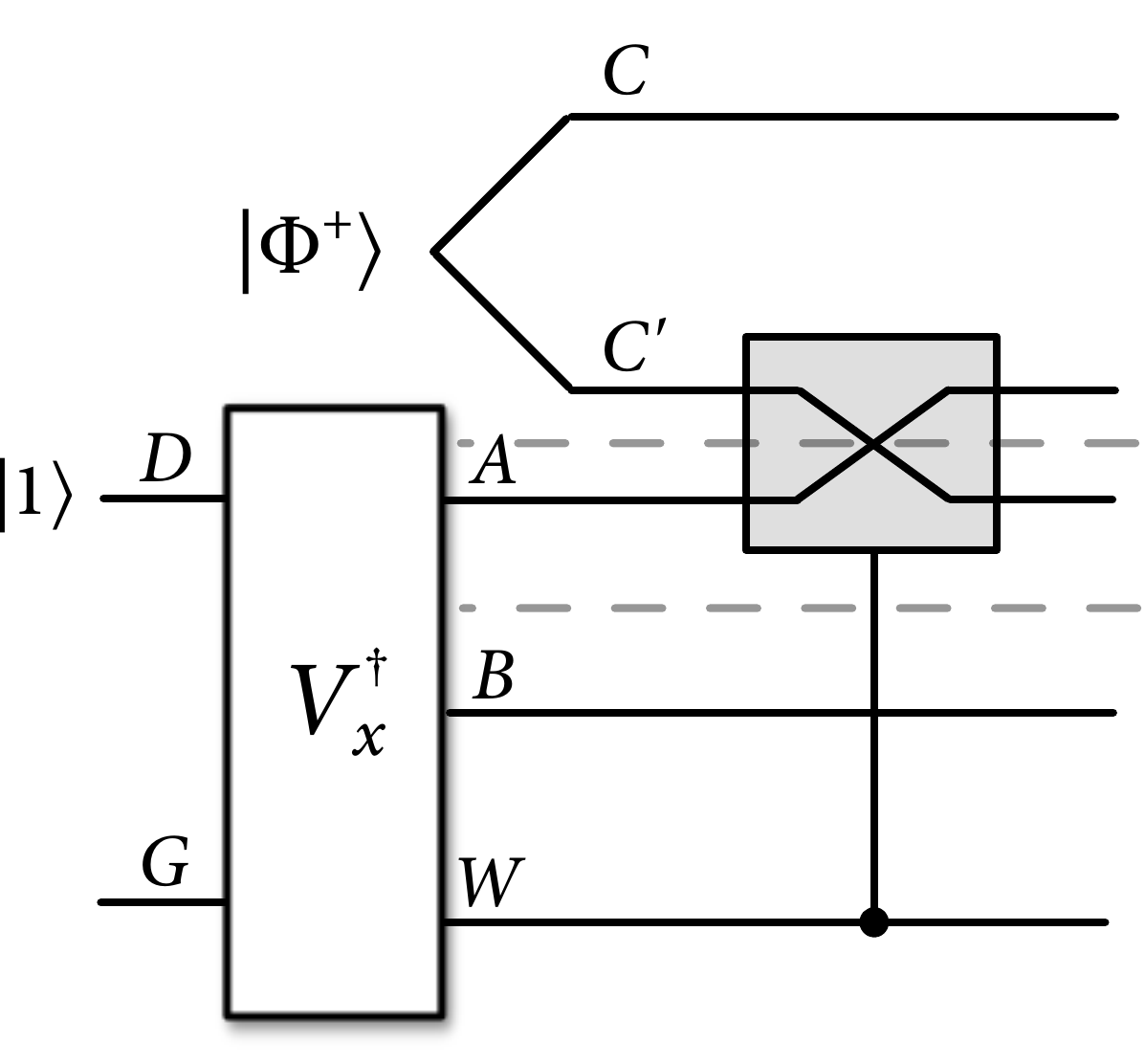}\\
    \end{tabular}
    \caption{
      \label{fig:qma2-reduction}
      (a) A unitary verifier circuit $V_x$ for an arbitrary verifier witnessing membership of $L$ in $\cls{QMA}(2)$ on input $x$.
      (b) The circuit $U$ produced by our reduction.  Dashed lines indicate that the output registers are to be divided along the bipartite cut $A:BCC'W$.}
    \end{center}
  \end{figure}
  The input registers $A,B$ are supplied by the two provers and the input register $W$ is a workspace register initialized to the $\ket{0}$ state.
  The output register $D$ is a decision qubit indicating acceptance or rejection of $x$ and the output register $G$ is a garbage register that consists of the remaining qubits upon which $V_x$ acts.

  In this proof we reduce the arbitrary problem $L$ to $(\alpha,\beta)$-\textsc{Bipartite Pure Product Isometry Output} where
  \begin{align}
    \alpha &= 2\sqrt{\delta} \, ,\\
    \beta &= 2\sqrt{1 - \Pa{\sqrt{\delta} + 2^{-n/2}}^2} \, ,
  \end{align}
  for any desired $n$.
  The desired hardness result then follows by an appropriate choice of $\delta,n$.

  The reduction is as follows.
  Given an instance $x$ of $L$ we produce a description of the following isometric circuit $U:G\to ABCC'W$:
  \begin{enumerate}
  \item Given the input register $G$, prepare a qubit $D$ in the accept state $\ket{\textnormal{yes}}$ and apply the inverse circuit $V_x^*$ to obtain registers $ABW$.
  \item Prepare registers $CC'$ in a $2n$-qubit maximally entangled state such as $n$ EPR pairs, which we denote by $\ket{\phi^+}$.
  \item Perform a unitary conditional swap gate that swaps registers $A$ and $C$ when $W$ is orthogonal to the $\ket{0}$ state and acts as the identity otherwise.
    (Here we implicitly pad the register $A$ with $\ket{0}$ qubits so as to have the same size as $C$.)
  \end{enumerate}
  See \expref{Figure}{fig:qma2-reduction}(b) for a graphical depiction of this circuit.

  Let us argue that this construction has the claimed properties.
  Suppose first that $x$ is a yes-instance of $L$ and let $\ket{\phi}_A\ket{\varphi}_B$ be a pure product state of registers $AB$ that causes the verifier to accept with high probability.
  That is, the state $V_x\ket{\phi}_A\ket{\varphi}_B\ket{0}_W$ has squared overlap at least $1-\delta$ with $\ket{\textnormal{yes}}\ket{\psi}$ for some state $\ket{\psi}$ of register $G$.
  Thus $U\ket{\psi}$ is $2\sqrt{\delta}$-close in trace distance to $\ket{\phi}_A\ket{\varphi}_B\ket{\phi^+}_{CC'}\ket{0}_W$, which is product with respect to the cut $A:BCC'W$, and so $U$ is a yes-instance of $(\alpha,\beta)$-\textsc{Bipartite Pure Product Isometry Output}.

  Next, suppose that $x$ is a no-instance of $L$.
  Fix any pure input state $\ket{\psi}$ for register $G$ and observe that
  \begin{equation} U\ket{\psi} = \Pi_0U\ket{\psi} + (I-\Pi_0)U\ket{\psi} \end{equation}
  where $\Pi_0=\kb{0}_W$ denotes the projection onto the $\ket{0}$ state for register $W$.
  From the definition of the circuit $U$ it is clear that we may write
  \begin{align}
    \Pi_0 U\ket{\psi}
    &= \Pi_0 \Pa{ V_x^*\ket{\textnormal{yes}}\ket{\psi} } \ot \ket{\phi^+}_{CC'} \\
    &= \ket{\zeta_{AB}}\ket{0}_W\ket{\phi^+}_{CC'} \label{eq:one-prime}\\
    (I-\Pi_0)U\ket{\psi}
    &= \operatorname{Swap}_{AC'} \Pa{ (I-\Pi_0) \Pa{ V_x^*\ket{\textnormal{yes}}\ket{\psi} } \ot \ket{\phi^+}_{CC'} } \\
    &= \ket{\xi_{BC'W}}\ket{\phi^+}_{AC}
  \end{align}
  for some choice of subnormalized pure states $\ket{\zeta_{AB}}$ and $\ket{\xi_{BC'W}}$ of registers $AB$ and $BC'W$, respectively.

  Then for any pure product state $\ket{\phi}$ of registers $A:BCC'W$ it holds that
  \begin{equation}
    \Abs{\bra{\phi}U\ket{\psi}}
    \leq \max_{\ket{\phi'}} \Abs{\bra{\phi'}\Pi_0 U \ket{\psi}} + \max_{\ket{\phi''}}\Abs{\bra{\phi''}(I-\Pi_0) U \ket{\psi}}  \end{equation}
  where the maxima on the right side are also over product states $\ket{\phi'},\ket{\phi''}$ of registers $A:BCC'W$.

  First, let us bound the maximum over $\ket{\phi'}$.
  It is clear from \eqref{eq:one-prime} that this maximum is achieved by some $\ket{\phi'}$ of the form
  \begin{equation} \ket{\phi'}=\ket{\phi}_A\ket{\varphi}_B\ket{0}_W\ket{\phi^+}_{CC'}, \end{equation}
  in which case we have
  \begin{equation} \Abs{\bra{\phi'}\Pi_0 U \ket{\psi}} = \Abs{ \bra{\phi}_A\bra{\varphi}_B\bra{0}_W V_x^*\ket{\textnormal{yes}}\ket{\psi} } \leq \sqrt{\delta} \end{equation}
  where the inequality follows from the assumption that $x$ is a no-instance of $L$.

  Next, let us bound the maximum over $\ket{\phi''}$.
  Since $(I-\Pi_0)U\ket{\psi}$ is maximally entangled on registers $AC$, its squared inner product with any product state $\ket{\phi''}$ is at most $2^{-n}$ as observed in \eqref{eq:fidelity-to-sep} of \expref{Section}{sec:prelim:owLOCC-to-sep}.

  We have thus shown that
  \begin{equation} \max_{\ket{\psi}} \max_{\textnormal{product $\ket{\phi}$}} \Abs{\bra{\phi}U\ket{\psi}}^2 \leq \Pa{\sqrt{\delta} + 2^{-n/2}}^2 \end{equation}
  and consequently
  \begin{equation} \min_{\ket{\psi}} \min_{\textnormal{product $\ket{\phi}$}} \Tnorm{U\psi U^* - \phi} \geq 2\sqrt{1 - \Pa{\sqrt{\delta} + 2^{-n/2}}^2}. \end{equation}
  We have thus shown that $U$ is a no-instance of $(\alpha,\beta)$-\textsc{Bipartite Pure Product Isometry Output}.
\end{proof}

\subsection{Equivalence of separability testing problems}
\label{sec:qma2-complete:equivalence}

We also consider two variants of \textsc{Pure Product Isometry Output} (\expref{Problem}{problem:product-isometry-output}) in which the task is to determine whether an isometry $U$ can be made to produce a (not necessarily pure) product state or a separable state.
Whereas \expref{Problem}{problem:product-isometry-output} restricts attention only to pure input states, in the following variants of the problem we also allow arbitrary mixed state inputs.
Formal specifications of these two variants of \expref{Problem}{problem:product-isometry-output} are given below.

\begin{problem}
  [$(\alpha,\beta,l)$-\textsc{Product Isometry Output}]
  \label{problem:mixed-product-isometry-output}
  \ \\[1mm]
  \begin{tabularx}{\textwidth}{lX}
    \emph{Input:} &
    A description of a quantum circuit that implements an isometry $U$ with an $l$-partite output system $A_1\cdots A_l$.
    \\[1mm]
    \emph{Yes:} &
    There is an input state $\rho$ such that $U\rho U^*$ is $\alpha$-close to a product state:
    \begin{equation} \min_\rho \min_{\sigma_1,\dots,\sigma_l} \Tnorm{ U\rho U^* - \sigma_1\ot\cdots\ot\sigma_l } \leq \alpha. \end{equation}
    \\
    \emph{No:} &
    For all input states $\rho$ it holds that $U\rho U^*$ is $\beta$-far from a product state:
    \begin{equation} \min_\rho \min_{\sigma_1,\dots,\sigma_l} \Tnorm{ U\rho U^* - \sigma_1\ot\cdots\ot\sigma_l } \geq \beta. \end{equation}
  \vspace{-\baselineskip}
  \end{tabularx}
\end{problem}

\begin{problem}
  [$(\alpha,\beta,l)$-\textsc{Separable Isometry Output}]
  \label{problem:separable-isometry-output}
  \ \\[1mm]
  \begin{tabularx}{\textwidth}{lX}
    \emph{Input:} &
    A description of a quantum circuit that implements an isometry $U$ with an $l$-partite output system $A_1\cdots A_l$.
    \\[1mm]
    \emph{Yes:} &
    There is an input state $\rho$ such that $U\rho U^*$ is $\alpha$-close to a separable state:
    \begin{equation} \min_\rho \min_{\sigma\in\mathcal{S}(A_1:\cdots:A_l)} \Tnorm{ U\rho U^* - \sigma } \leq \alpha. \end{equation}
    \\
    \emph{No:} &
    For all input states $\rho$ it holds that $U\rho U^*$ is $\beta$-far from separable:
    \begin{equation} \min_\rho \min_{\sigma\in\mathcal{S}(A_1:\cdots:A_l)} \Tnorm{ U\rho U^* - \sigma } \geq \beta. \end{equation}
  \vspace{-\baselineskip}
  \end{tabularx}
\end{problem}

We now argue that, for each $l$, these problems are equivalent to one another for a wide range of choices of $(\alpha,\beta)$.
These equivalences are corollaries of the following proposition, which relates minimal distance from separable to minimal distance from pure product.

\begin{proposition}[Separable-to-pure product reduction]
\label{lem:pure-to-sep}

  Let $U$ be an isometry with an $l$-partite output system $A_1\cdots A_l$ and suppose that there is an input state $\rho$ such that $U\rho U^*$ is $\delta$-close to some separable state $\sigma\in\mathcal{S}(A_1:\cdots:A_l)$:
  \begin{equation} \Tnorm{U\rho U^* - \sigma} \leq \delta. \end{equation}
  Then there is a pure input state $\ket{\psi}$ such that $U\psi U^*$ is $4\sqrt{\delta}$-close to some pure product state $\ket{\phi_1}\ot\cdots\ot\ket{\phi_l}$:
  \begin{equation} \Tnorm{U\psi U^* - \phi_1\ot\cdots\ot\phi_l} \leq 4\sqrt{\delta}. \end{equation}

\end{proposition}

\begin{proof}
Let
\begin{equation} \sigma = \sum_x p_x \phi_1^x\ot\cdots\ot\phi_l^x \end{equation}
be a decomposition of $\sigma$ as a probabilistic mixture of pure product states and let
\begin{equation} \ket{\zeta} = \sum_x \sqrt{p_x} \ket{x}_R \ot \ket{\phi_1^x}\ot\cdots\ot\ket{\phi_l^x} \end{equation}
be a purification of $\sigma$ on registers $RA_1\cdots A_l$.

Let $S$ denote the input register for $U$.
It follows from \eqref{eq:FvG-ineqs} and Uhlmann's Theorem that there is a purification $\ket{\psi}$ of $\rho$ on registers $RS$ with
\begin{equation} \Tnorm{ U\psi U^*-\zeta} \leq 2\sqrt{\delta}. \end{equation}
Write $\ket{\psi}$ as
\begin{equation} \ket{\psi} = \sum_{x}\sqrt{q_x} \ket{x}_{R}\ot \ket{\psi^x} \end{equation}
for some probability vector $q$ and states $\set{\ket{\psi^x}}_x$ (not necessarily orthogonal).
Apply a dephasing channel in the basis $\set{\ket{x}}_x$ on register $R$ and use contractivity of trace norm under quantum channels to obtain
\begin{equation}
  \Tnorm{ U\psi U^*-\zeta}
  \geq \Tnorm{\sum_x q_x \kb{x} \ot U\psi^x U^* - \sum_x p_x \kb{x}\ot\phi_1^x\ot\cdots\ot\phi_l^x}
\end{equation}
Combining this bound with the triangle inequality, we have
\begin{align}
  & \sum_x p_x \Tnorm{U\psi^x U^* - \phi_1^x\ot\cdots\ot\phi_l^x} \\
  ={}& \Tnorm{ \sum_x p_x \kb{x} \ot \Pa{U\psi^x U^* - \phi_1^x\ot\cdots\ot\phi_l^x} } \\
  \leq{}& \Tnorm{ \sum_x q_x \kb{x} \ot U\psi^x U^* - \sum_x p_x \kb{x} \ot U\psi^x U^* } \\
    &{}+ \Tnorm{ \sum_x q_x \kb{x} \ot U\psi^x U^* - \sum_x p_x \kb{x} \ot \phi_1^x\ot\cdots\ot\phi_l^x }\\
  \leq{}& 4\sqrt{\delta}.
\end{align}
Since this inequality holds for a convex combination over terms indexed by $x$, it must also hold for at least one choice of $\psi^x,\phi_1^x,\dots,\phi_l^x$.
\end{proof}

\begin{corollary}[Equivalence of problems] \label{cor:equiv-isom-probs}
  The following hold for all $l$ and all $\alpha<\beta$:
  \begin{enumerate}

  \item \label{it:mixed-leq-pure}

    Both $(\alpha,\beta,l)$-\textsc{Product Isometry Output} and $(\alpha,\beta,l)$-\textsc{Separable Isometry Output} trivially reduce to $(4\sqrt{\alpha},\beta,l)$-\textsc{Pure Product Isometry Output}.

  \item \label{it:pure-leq-mixed}

    Conversely, $(\alpha,\beta,l)$-\textsc{Pure Product Isometry Output} trivially reduces to both $(\alpha,\beta^2/16,l)$-\textsc{Product Isometry Output} and $(\alpha,\beta^2/16,l)$-\textsc{Separable Isometry Output}.

  \end{enumerate}
\end{corollary}

\begin{proof}
  By definition, no-instances of both $(\alpha,\beta,l)$-\textsc{Product Isometry Output} and $(\alpha,\beta,l)$-\textsc{Separable Isometry Output} are also no-instances of $(4\sqrt{\alpha},\beta,l)$-\textsc{Pure Product Isometry Output}.
  By \expref{Proposition}{lem:pure-to-sep}, yes-instances of both $(\alpha,\beta,l)$-\textsc{Product Isometry Output} and $(\alpha,\beta,l)$-\textsc{Separable Isometry Output} are also yes-instances of $(4\sqrt{\alpha},\beta,l)$-\textsc{Pure Product Isometry Output}.

  By definition, yes-instances of $(\alpha,\beta,l)$-\textsc{Pure Product Isometry Output} are also yes-instances of both $(\alpha,\beta^2/16,l)$-\textsc{Product Isometry Output} and $(\alpha,\beta^2/16,l)$-\textsc{Separable Isometry Output}.
  By the contrapositive of \expref{Proposition}{lem:pure-to-sep}, no-instances of $(\alpha,\beta,l)$-\textsc{Pure Product Isometry Output} are also no-instances of both $(\alpha,\beta^2/16,l)$-\textsc{Product Isometry Output} and $(\alpha,\beta^2/16,l)$-\textsc{Separable Isometry Output}.
\end{proof}

\begin{corollary}[$\cls{QMA}(2)$-completeness of equivalent problems]
\label{cor:qma2-complete}

  The following hold:
  \begin{enumerate}
  \item \expref{Problems}{problem:mixed-product-isometry-output} and \ref{problem:separable-isometry-output} are in $\cls{QMA}(2)$ for all $l$ and all $\alpha<\beta^2/16$.
  \item These two problems are $\cls{QMA}(2)$-hard for all $l\geq 2$ and all $(\alpha,\beta)=(\varepsilon,1/4-\varepsilon)$, even when $\varepsilon$ decays exponentially in the input length.
  \end{enumerate}
  Thus, these two problems are $\cls{QMA}(2)$-complete for all $l\geq 2$ if both $0<\alpha<\beta^2/16$ and $\beta<1/4$.

\end{corollary}

\begin{remark}
  The fact that \expref{Problems}{problem:mixed-product-isometry-output} and \ref{problem:separable-isometry-output} are $\cls{QMA}(2)$-hard only for $(\alpha,\beta)=(\varepsilon,1/4-\varepsilon)$ instead of the best possible $(\varepsilon,2-\varepsilon)$ is an artifact of \expref{Proposition}{lem:pure-to-sep}.
  The best possible hardness result would be obtained if the bound in \expref{Proposition}{lem:pure-to-sep} could somehow be improved from $4\sqrt{\delta}$ to $\sqrt{2\delta}$.
\end{remark}

\section{\textsc{Product State} is $\cls{QSZK}$-complete} \label{sec:qszk-completeness}

In this section we prove $\cls{QSZK}$-completeness of the problem of determining whether the state prepared by a given quantum circuit is close to a product state.

\begin{problem}
  [$(\alpha,\beta)$-\textsc{Product State}]
  \label{problem:product-state}
  \ \\[1mm]
  \begin{tabularx}{\textwidth}{lX}
    \emph{Input:} &
    A description of a quantum circuit that prepares an $l$-partite mixed state $\rho$.
    \\[1mm]
    \emph{Yes:} &
    $\rho$ is $\alpha$-close to a product state:
    \begin{equation} \min_\rho \min_{\sigma_1,\dots,\sigma_l} \Tnorm{ \rho - \sigma_1\ot\cdots\ot\sigma_l } \leq \alpha. \end{equation}
    \\
    \emph{No:} &
    $\rho$ is $\beta$-far from product:
    \begin{equation} \min_\rho \min_{\sigma_1,\dots,\sigma_l} \Tnorm{ \rho - \sigma_1\ot\cdots\ot\sigma_l } \geq \beta. \end{equation}
  \vspace{-\baselineskip}
  \end{tabularx}
\end{problem}

The main result of this section is the following theorem:

\begin{theorem}[\textsc{Product State} is $\cls{QSZK}$-complete] \label{thm:qszk-completeness}

  The following hold:
  \begin{enumerate}
  \item $(\alpha,\beta,l)$-\textsc{Product State} is in $\cls{QSZK}$ for all $l$ and all $\alpha < \beta^2/(l+1)$.
  \item $(\varepsilon,2-\varepsilon)$-\textsc{Bipartite Product State} is $\cls{QSZK}$-hard, even when $\varepsilon$ decays exponentially in the input length.
  \end{enumerate}
  Thus, the problem is $\cls{QSZK}$-complete for all $l\geq 2$ and all $0<\alpha<\beta^2/(l+1)$ and $\beta<2$.

\end{theorem}

This result is proven by establishing equivalence between the \textsc{Product State} problem and the \textsc{Quantum State Similarity} problem, which is defined as follows:

\begin{problem}
  [$(\alpha,\beta)$-\textsc{Quantum State Similarity}]
  \label{problem:co-qsd}
  \ \\[1mm]
  \begin{tabularx}{\textwidth}{lX}
    \emph{Input:} &
    Descriptions of two quantum circuits that prepare mixed states $\rho_0,\rho_1$.
    \\
    \emph{Yes:} &
    $\rho_0$ and $\rho_1$ are $\alpha$-close:
    \( \Tnorm{ \rho_0 - \rho_1 } \leq \alpha. \)
    \\
    \emph{No:} &
    $\rho_0$ and $\rho_1$ are $\beta$-far apart:
    \( \Tnorm{ \rho_0 - \rho_1 } \geq \beta. \)
  \vspace{-\baselineskip}
  \end{tabularx}
\end{problem}

\expref{Problem}{problem:co-qsd} is known to be $\cls{QSZK}$-complete.
Specifically, $(\alpha,\beta)$-\textsc{Quantum State Similarity} is contained in $\cls{QSZK}$ for all $\alpha<\beta^2$ and $(\varepsilon,2-\varepsilon)$-\textsc{Quantum State Similarity} is $\cls{QSZK}$-hard, even when $\varepsilon$ decays exponentially in the input length \cite{W02,W09zkqa}.
Thus, \expref{Theorem}{thm:qszk-completeness} can be proved by reducing \expref{Problems}{problem:product-state} and \ref{problem:co-qsd} to each other.

\subsection{Containment in $\cls{QSZK}$}

Our reduction from \textsc{Product State} to \textsc{Quantum State Similarity} employs the fact that if $\rho$ is close to a product state then $\rho$ is also close to the product of its reduced states.
We are not aware of an explicit proof of this fact in the literature, so we provide a proof.

\begin{lemma}[Approximation by a product of reduced states]
  \label{lm:reduced-product}
  Let $\rho$ be a state of registers $A_1,\dots,A_l$ and suppose there is a product state $\sigma_1\ot\cdots\ot\sigma_l$ with
  \begin{equation} \Tnorm{ \rho - \sigma_1\ot\cdots\ot\sigma_l } \leq \alpha. \end{equation}
  Then it follows that
  \begin{equation} \Tnorm{ \rho - \rho_{A_1}\ot\cdots\ot\rho_{A_l} } \leq (l+1)\alpha \end{equation}
  where $\rho_{A_i}$ denotes the reduced state of $\rho$ on register $A_i$ for $i=1,\dots,l$.
\end{lemma}

\begin{proof}
By the triangle inequality we have
\begin{equation}
  \Tnorm{\rho - \rho_{A_1}\ot\cdots\ot\rho_{A_l}} \leq \Tnorm{\rho - \sigma_1\ot\cdots\ot\sigma_l} + \Tnorm{\sigma_1\ot\cdots\ot\sigma_l - \rho_{A_1}\ot\cdots\ot\rho_{A_l} }.
\end{equation}
By assumption the first term on the right is no larger than $\alpha$.
For the second term, another application of the triangle inequality yields
\begin{align}
  & \Tnorm{ \sigma_1\ot\cdots\ot\sigma_l - \rho_{A_1}\ot\cdots\ot\rho_{A_l} } \label{eq:reduced-product-1}\\
  \leq{} & \Tnorm{ \sigma_1\ot\cdots\ot\sigma_l - \rho_{A_1}\ot\sigma_2\ot\cdots\ot\sigma_l } + \Tnorm{ \rho_{A_1}\ot\sigma_2\ot\cdots\ot\sigma_l - \rho_{A_1}\ot\cdots\ot\rho_{A_l} } \\
  ={}& \Tnorm{ \sigma_1 - \rho_{A_1} } + \Tnorm{ \sigma_2\ot\cdots\ot\sigma_l - \rho_{A_2}\ot\cdots\ot\rho_{A_l} \label{eq:reduced-product}}
\end{align}
By the contractivity of the trace norm under partial trace we have
\begin{equation} \Tnorm{ \sigma_i - \rho_{A_i} } \leq \Tnorm{ \sigma_1\ot\cdots\ot\sigma_l - \rho } \leq \alpha \end{equation}
for each $i=1,\dots,l$.
The lemma then follows by applying \eqref{eq:reduced-product-1}-\eqref{eq:reduced-product} inductively.
\end{proof}

We are now ready to reduce \textsc{Product State} to \textsc{Quantum State Similarity}.

\begin{proposition}
  $(\alpha,\beta,l)$-\textsc{Product State} is in $\cls{QSZK}$ for all $l$ and all $\alpha<\beta^2/(l+1)$.
\end{proposition}

\begin{proof}
We reduce $(\alpha,\beta,l)$-\textsc{Product State} to $((l+1)\alpha,\beta)$-\textsc{Quantum State Similarity}.
It then follows that $(\alpha,\beta,l)$-\textsc{Product State} is in $\cls{QSZK}$ whenever $\alpha<\beta^2/(l+1)$ as desired.

The reduction is as follows:
given an instance $\rho$ of $(\alpha,\beta,l)$-\textsc{Product State}, one can construct circuits that prepare states
\begin{align}
  \rho_0 &= \rho \\
  \rho_1 &= \rho_{A_1}\ot\cdots\ot\rho_{A_l}.
\end{align}
Specifically, we use the original circuit to make $\rho_0 = \rho$, and we use the original circuit $l$ times to make $l$ copies of $\rho$ and then trace over the appropriate subsystems to make
$\rho_1 = \rho_{A_1}\ot\cdots\ot\rho_{A_l}$.
If $\rho$ is a yes-instance of $(\alpha,\beta,l)$-\textsc{Product State} then by \expref{Lemma}{lm:reduced-product} we have that $\tnorm{\rho_0-\rho_1}\leq (l+1)\alpha$.
Conversely, if $\rho$ is a no-instance of $(\alpha,\beta,l)$-\textsc{Product State} then it must be that $\tnorm{\rho_0-\rho_1}\geq\beta$.
\end{proof}

\subsection{Hardness for $\cls{QSZK}$}

\begin{proposition}
  $(\varepsilon,2-\varepsilon)$-\textsc{Bipartite Product State} is $\cls{QSZK}$-hard, even when $\varepsilon$ decays exponentially in the input length.
\end{proposition}

\begin{proof}
  We reduce $(\delta,2-\delta)$-\textsc{Quantum State Similarity} to $(\alpha,\beta)$-\textsc{Bipartite Product State} for
  \begin{align}
    \alpha &= n\delta/2 \, ,\\
    \beta &= 2-2^{-\Omega(n)} \, ,
  \end{align}
  for any desired $n$.
  The desired hardness result then follows by an appropriate choice of $\delta,n$.

  The reduction is as follows.
  Given an instance $(\rho_0,\rho_1)$ of \textsc{Quantum State Similarity} we construct a circuit that prepares $n$ copies of the bipartite state $\omega_{A:S}$ of registers $AS$ given by
  \begin{equation} \omega_{A:S} = \frac{1}{2}\kb{0}_A\ot\rho_0 + \frac{1}{2}\kb{1}_A\ot\rho_1. \end{equation}
  \expref{Figure}{fig:qszk-reduction} illustrates an isometric circuit for preparing (a purification of) a single copy of $\omega_{A:S}$.
  \begin{figure}
  \begin{center}
  \includegraphics[scale=.6]
  {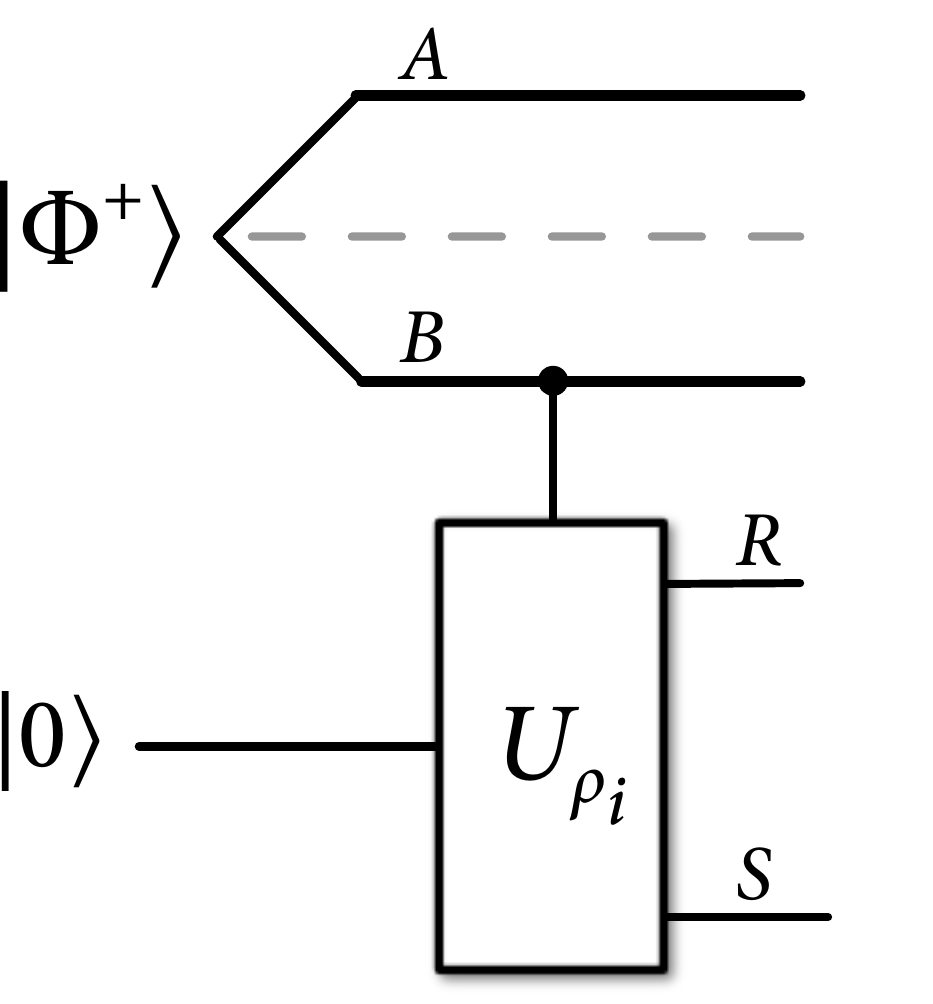}
  \caption{\label{fig:qszk-reduction}
    The isometric circuit that prepares a purification of $\omega_{A:S}$.
    This circuit is produced by our reduction from \textsc{Quantum State Similarity} to \textsc{Bipartite Product State}.
    Here $U_{\rho_i}$ are the unitary circuits that prepare $\rho_i=\ptr{R}{U_{\rho_i}\kb{0}U_{\rho_i}^*}$ in register $S$ for $i\in\set{0,1}$.
    This same circuit is also produced by the reduction of \cite{HMW13,HMW14} from \textsc{Quantum State Distinguishability} to the one-way LOCC version of \textsc{Separable State}, except that reduction discards register $S$ instead of $R$.
  }
  \end{center}
  \end{figure}

  If $(\rho_0,\rho_1)$ is a yes-instance of $(\delta,2-\delta)$-\textsc{Quantum State Similarity} then the trace distance between $\omega_{A:S}$ and the product state
  \begin{equation} \sigma_{A:S} = \frac{1}{2}\Pa{\kb{0}+\kb{1}}_A\ot\rho_0 \end{equation}
  is at most $\frac{1}{2}\tnorm{\rho_0-\rho_1}\leq\delta/2$.
  It then follows from \cite[Lemma 8]{W02} that
  \begin{equation} \Tnorm{\omega_{A:S}^{\ot n} - \sigma_{A:S}^{\ot n}} \leq n\delta/2. \end{equation}
  As $\sigma_{A:S}^{\ot n}$ is product relative to the bipartite cut $A_1\cdots A_n:S_1\cdots S_n$ it must be that $\omega_{A:S}^{\ot n}$ is a yes-instance of $(\alpha,\beta)$-\textsc{Bipartite Product State} relative to this cut.

  By contrast, if $(\rho_0,\rho_1)$ is a no-instance of $(\delta,2-\delta)$-\textsc{Quantum State Similarity} then $\omega_{A:S}$ is almost perfectly correlated on $A:S$ and hence far from product.
Recall that trace
distance is equal to the maximum probability of distinguishing states over all possible measurements,
so we can lower bound the distance to the nearest product state by considering a particular protocol
to distinguish $\omega_{A:S}$ from any product state. In this protocol, we begin
by measuring the first
qubit (register $A$) in the computational basis and by performing the Helstrom measurement $\{\Pi_{0},\Pi_{1}\}$
on the second system, storing the two measurement outcomes in classical registers.

It is straightforward to calculate the state $\omega'_{A:S'}$ that results
after applying the protocol above to the state
$\omega_{A:S}$:
\begin{equation}
\omega'_{A:S'} = \frac{1}{2}\text{Tr}\{\Pi_0 \rho_0\}\ketbra{00}{00} +
\frac{1}{2}\text{Tr}\{\Pi_1 \rho_1\}\ketbra{11}{11} +
\frac{1}{2}\text{Tr}\{\Pi_0 \rho_1\} \ketbra{10}{10} +
\frac{1}{2}\text{Tr}\{\Pi_1 \rho_0\}\ketbra{01}{01}.
\end{equation}
Recall that the Helstrom measurement distinguishes two states $\rho_0$ and $\rho_1$
with the following success probability:
\begin{equation} \frac{1}{2}\text{Tr}\{\Pi_0 \rho_0\} + \frac{1}{2}\text{Tr}\{\Pi_1 \rho_1\}
=\frac{1}{2}\left(1 + \frac{1}{2}\norm{\rho_0 - \rho_1}_1\right),\end{equation}
and the following error probability:
\begin{equation} \frac{1}{2}\text{Tr}\{\Pi_0 \rho_1\} + \frac{1}{2}\text{Tr}\{\Pi_1 \rho_0\}
= \frac{1}{2}\left(1 - \frac{1}{2}\norm{\rho_0 - \rho_1}_1\right) .\end{equation}
Using this fact, it is straightforward to establish that the trace distance between
$\omega'_{A:S'}$ and the perfectly correlated state $\overline{\Phi}_{A:S'}$, defined as
\begin{equation}\overline{\Phi}_{A:S'} \equiv \frac{1}{2}(\ketbra{00}{00} + \ketbra{11}{11}),\end{equation}
is no larger than
\begin{equation}1 - \frac{1}{2}\norm{\rho_0 - \rho_1}_1 \leq \frac{\delta}{2} .\end{equation}

For a product state, the two measurement outcomes
must be uncorrelated, and so we can write the
result of applying the above protocol to any product state using
the probability~$p$ of measuring $\ketbra{0}{0}$
and the probability $q$ of measuring $\Pi_0$:
\begin{align}
\sigma_{p,q} &= pq \ketbra{00}{00} + p(1-q) \ketbra{01}{01} + q(1-p)\ketbra{10}{10} + (1-p)(1-q)\ketbra{11}{11}.
\end{align}
From the monotonocity of trace distance under quantum operations, it follows that
\begin{equation}
\min_{\sigma_0,\sigma_1} \norm{ \sigma_0 \otimes \sigma_1 - \omega_{A:S} }_1
\geq
\min_{p,q} \norm{ \sigma_{p,q} - \omega'_{A:S} }_1
\end{equation}
Due to symmetry, we can take $p \le \frac{1}{2}$
without loss of generality. We can then bound the minimum distance
of $\sigma_{p,q}$ to $\omega'_{A:S}$:
\begin{align}
\min_{p,q} \norm{ \sigma_{p,q} - \omega'_{A:S} }_1 &\ge \min_{p,q}
\norm{\sigma_{p,q} -  \overline{\Phi}_{A:S}  }_1 -
\norm{\overline{\Phi}_{A:S} - \omega'_{A:S}}_1 \\
&\ge \norm{\sigma_{p,q} - \overline{\Phi}_{A:S}}_1 - \frac{\delta}{2}\\
&= \av{\frac{1}{2} - pq} + \av{\frac{1}{2} - (1-p)(1-q)} + \av{p(1-q)} + \av{q(1-p)} - \frac{\delta}{2} \\
&= \frac{1}{2} - pq +
\av{\frac{1}{2} - (1-p)(1-q)}
+ p(1-q) + q(1-p) - \frac{\delta}{2} \\
& \ge \frac{1}{2} - pq
+ p(1-q) + q(1-p) - \frac{\delta}{2} \\
&\ge \frac{1}{2} + p(1- q) - \frac{\delta}{2}\\
&\ge \frac{1 - \delta}{2},
\end{align}
where the first line follows from the triangle inequality, and the fourth through last lines follow from the fact that
$0 \le p \le \frac{1}{2}$ and $0 \le q \le 1$.
  It then follows from \cite[Lemma 8]{W02} that for suitably high $n$ we have that $\omega_{A:S}^{\ot n}$ is at least $\Pa{2-2^{-\Omega(n)}}$-far from any product state and so this state is a no-instance of $(\alpha,\beta)$-\textsc{Bipartite Product State} as desired.
\end{proof}

\begin{remark}
\expref{Theorem}{thm:qszk-completeness} provides a different proof that the promise problem \textsc{Error Correctability} of \cite{HH13} is $\cls{QSZK}$-complete (with a proof preceding this one given in \cite{HS13}).
Indeed, \textsc{Error Correctability} is the task of deciding whether it is possible to decode a maximally entangled state from systems $R$ and $B$ when a unitary specified as a quantum circuit acts on systems $R$, $B$, and $E$, such that systems $R$ and $B$ are initialized to the maximally entangled state and system $E$ is initialized to the all-zero state.
In this problem, there is a promise that it is either possible to decode maximal entanglement (approximately) or impossible to do so.
Due to the ``decoupling theorem'' often used in quantum information theory \cite{qcap2008first}, the question of whether it is possible to decode maximal entanglement between systems $R$ and $B$ is equivalent to the question of whether systems $R$ and $E$ are in a product state.
Thus, it follows from \expref{Theorem}{thm:qszk-completeness} that \textsc{Error Correctability} and \textsc{Product State} are reducible to each another and that \textsc{Error Correctability} is $\cls{QSZK}$-complete.
\end{remark}


\section{A short quantum game for the one-way LOCC version of \textsc{Separable State}}
\label{sec:sqg-inclusion}

In \cite{HMW13,HMW14} it was shown that the one-way LOCC version of the \textsc{Separable State} problem admits a two-message quantum interactive proof, so that the problem lies inside $\cls{QIP}(2)$.
In this section we show that this problem also admits a short quantum game, putting it inside $\cls{SQG}$, too.
As mentioned in \expref{Section}{sec:overview}, this result is not a complexity-theoretic improvement over prior work.
But it is interesting that the one-way LOCC version of \textsc{Separable State} admits a natural, single-message quantum proof provided that the verifier has help from a second competing prover.
Recall the definition of the one-way LOCC version of the \textsc{Separable State} problem \cite{HMW13,HMW14}:

\begin{problem}
  [$(\alpha,\beta,l)$-\textsc{Separable State}, one-way LOCC version]
  \label{problem:sep-state-1locc}
  \ \\[1mm]
  \begin{tabularx}{\textwidth}{lX}
    \emph{Input:} &
    A description of a quantum circuit that prepares a state $\rho$ of registers $A_1\cdots A_l$.
    \\[1mm]
    \emph{Yes:} &
    $\rho$ is $\alpha$-close in trace distance to a separable state:
    \begin{equation} \min_{\sigma\in\mathcal{S}(A_1:\cdots:A_l)} \Tnorm{\rho-\sigma} \leq \alpha. \end{equation}
    \vspace{-\baselineskip}
    \\
    \emph{No:} &
    $\rho$ is $\beta$-far in one-way LOCC distance from separable:
    \begin{equation} \min_{\sigma\in\mathcal{S}(A_1:\cdots:A_l)} \Lnorm{\rho-\sigma} \geq \beta. \end{equation}
  \vspace{-\baselineskip}
  \end{tabularx}
\end{problem}

The main result of this section is the following proposition:

\begin{proposition} \label{thm:co-QSEP-in-SQG}
  The one-way LOCC version of $(\alpha,\beta,l)$-\textsc{Separable State} is in $\cls{SQG}$ for all $l$ and all $\alpha<\beta$.
\end{proposition}

\begin{proof}
Suppose that registers $A_1\cdots A_l$ have combined total dimension $D$.
The verifier witnessing membership of the problem in $\cls{SQG}$ is described as follows:
\begin{enumerate}

\item \label{it:yes-first}
  Receive $kl$ registers from the yes-prover labeled $A_i^j$ for $i=1,\dots,l$ and $j=1,\dots,k$ where
  \begin{equation} \label{eq:complement-k}
    k = \left\lceil l + \frac{16l^2\log D}{(\beta-\alpha)^2} \right\rceil.
  \end{equation}
  (Intuitively, these registers contain a purported $k$-extension of $\rho$.)

\item \label{it:perm-test}
  Perform $l$ permutation tests: one for each group $(A_i^1,\dots,A_i^k)$ of $k$ registers.
  Reject immediately if any test fails.
  Discard all registers except $A_1^1,\dots,A_l^1$, letting $\sigma$ denote the reduced state of these remaining registers.

\item
  Prepare a copy of $\rho$ using the input circuit and choose a random bit $b\in\{0,1\}$.
  If $b=0$ then send $\rho$ to the no-prover.
  Otherwise, send $\sigma$ to the no-prover.
  (Intuitively, the no-prover is challenged to identify whether the state he receives from the verifier is $\rho$ or $\sigma$.)

\item
  Receive a single bit $b'$ from the no-prover.
  Reject if and only if $b'=b$.

\end{enumerate}
Let us argue that this protocol is correct.
For yes-instances an optimal strategy for the yes-prover is to select a separable state $\sigma$ that is $\alpha$-close in trace distance to $\rho$ and send the verifier a $k$-extension of $\sigma$.
As $\sigma$ is separable, such an extension must exist for every choice of~$k$ and so the permutation test passes with certainty.
The no-prover is then faced with the task of distinguishing $\sigma$ from $\rho$, which he can do with probability no larger than $1/2+\alpha/4$, implying that the verifier accepts with probability at least $1/2-\alpha/4$.

For no-instances an optimal strategy for the no-prover is to perform a measurement that distinguishes $\rho$ from the convex set $\mathcal{E}_k$ of $k$-extendible states with probability at least
\begin{equation} \label{eq:preject-sqg}
  \frac{1}{2} + \frac{1}{4} \min_{\sigma \in \mathcal{E}_k} \tnorm{\rho - \sigma}.
\end{equation}
(The existence of such a measurement was first shown in \cite{GW05} and a simple proof can be found in Yu, Duan, and Xu \cite{YDX12}.)

To see that the yes-prover cannot win, observe that if the permutation test of step \ref{it:perm-test} passes then the state of all $kl$ registers $A_i^j$ received from the yes-prover is projected into the symmetric subspaces of $(A_i^1,\dots,A_i^k)$ for each $i=1,\dots,l$.
The set of such states is contained in the set $\mathcal{E}_k$ of $k$-extendible states, and we know from \expref{Theorem}{lem:multi-k-ext} and our choice of $k$ that
\begin{equation} \min_{\sigma \in \mathcal{E}_k} \tnorm{\rho - \sigma} \geq \frac{\beta+\alpha}{2}. \end{equation}
Thus, the no-prover convinces the verifier to reject with probability at least $1/2+(\beta+\alpha)/8$, implying that the verifier accepts with probability at most $1/2-(\beta+\alpha)/8$.
This protocol witnesses membership in $\cls{SQG}$ whenever
\(  1/2-\alpha/4 > 1/2-(\beta+\alpha)/8, \)
which occurs whenever $\alpha<\beta$.
\end{proof}

\section{Operational interpretations of geometric measures of entanglement}
\label{sec:geometric-measure}

Our work has a close connection to several entanglement measures known collectively as the \emph{geometric measure of entanglement}---see \cite{WG03,CAH13} and references therein.
This is also the case with the work in \cite{HM10} and we comment briefly on this connection.
The original definition of the geometric measure of entanglement for a pure state $\ket{\psi}$ of registers $AB$ is defined as the maximum squared overlap with a pure product state:
\begin{equation} \label{eq:geometric-pure}
  \max_{\ket{\phi}_A, \ket{\varphi}_B} \Abs{ \braket{\phi \otimes \varphi}{\psi}}^2.
\end{equation}
This quantity has an operational interpretation as the maximum probability with which the state $\ket{\psi}$ would pass a test for being a pure product state.
By taking the negative logarithm one obtains an entropic-like quantity that is equal to the geometric measure of entanglement and satisfies a list of desirable requirements that any entanglement measure ought to meet.

If one has a promise that the quantity \eqref{eq:geometric-pure} is larger than $1-\varepsilon$ or smaller than $\varepsilon$ (as in the definition of the \textsc{Pure Product State} problem, \expref{Problem}{problem:pure-product-state}) then the product test can be used to determine which is the case.
However, this observation does not directly give an operational interpretation of the quantity in \eqref{eq:geometric-pure}.
Rather, an operational interpretation of \eqref{eq:geometric-pure} is given by the quantum interactive proof in \cite{HMW13,HMW14} for \textsc{Separable State}, whose maximum acceptance probability for a given state $\rho$ of registers $AB$ is given by
\begin{equation} \label{eq:geometric-sep}
  \max_{\sigma \in \mathcal{S}(A:B)} F(\rho_{AB}, \sigma_{AB})\, ,
\end{equation}
of which \eqref{eq:geometric-pure} is a special case when $\rho_{AB}$ is pure.
(This bound holds in the limit of large $k$, the number of registers sent by the prover in a purported $k$-extension of $\rho$.)
The operational interpretation for \eqref{eq:geometric-sep} is that it is the maximum probability with which a prover could convince a verifier that a state $\rho$ is separable by acting on a purification of $\rho$.

%

Our work has also unveiled and provided operational interpretations for
other quantifiers of entanglement
that fall within the geometric class. Indeed, the maximum acceptance probability of our quantum witness for the one-way LOCC version of \textsc{Separable Isometry Output} is bounded by
\begin{equation}
\max_{\rho, \, \sigma_{AB} \in \mathcal{S}}
F(U(\rho_S \otimes \ketbra{0}{0})U^\dagger,\sigma_{AB}),
\end{equation}
again a bound that holds in the large $k$ limit. Clearly, this quantity is related to
the so-called ``entangling power'' of the unitary $U$ \cite{ZZF00}, that is, its ability to take
a product state input to an entangled output no matter what the input is. Furthermore,
the quantum interactive proof for the one-way LOCC version of \textsc{Separable Channel Output} given in \cite{HMW13,HMW14} has the following upper bound on its maximum acceptance probability:
\begin{equation}
\max_{\rho, \, \sigma_{AB} \in \mathcal{S}}
F(\mathcal{N}_{S\to AB}(\rho_S),\sigma_{AB}),
\end{equation}
where $\mathcal{N}_{S\to AB}$ is a quantum channel with input system $S$
and output systems $AB$. Again, this bound holds in the limit of large $k$.
The above measure is related to the entangling capabilities of a quantum channel
no matter what the input is, and the quantum interactive proof provides an operational interpretation
for the above quantity as well.

\section{Discussion: Does nondeterminism trump the one-way LOCC distance?}

An interesting and surprising comparison emerges in light of the combined results of the present paper with those of \cite{HMW13,HMW14}.
For isometric channels, it is no surprise that detecting product outputs is easier than detecting separable outputs when no-instances in the former problem are promised to be far from product in one-way LOCC distance instead of trace distance:
these problems are complete for $\cls{QMA}$ and $\cls{QMA}(2)$, respectively.
For states, however, detecting separability is \emph{harder} than detecting productness, \emph{even when no-instances in the former problem are promised to be far from separable in one-way LOCC distance:}
the former is both $\cls{QSZK}$- and $\cls{NP}$-hard while the latter is $\cls{QSZK}$-complete.

An anonymous reviewer suggests one possible explanation for this phenomenon: the added difficulty of nondeterminism trumps the reduced difficulty of the one-way LOCC promise.
Specifically, detecting entangled or correlated isometry outputs is inherently ``nondeterministic,'' as one must guess the proper input to the isometry.
Similarly, detecting separable states is also ``nondeterministic,'' as one must guess a mixture of product states.
By contrast, product states have no nondeterminism of this form and so we can expect the corresponding detection problem to be easier, even when one demands a lower error tolerance via the trace distance.

This explanation is interesting and intuitive.
To this explanation we add the following observation: even product states contain ``nondeterminism'' in the sense that we must also recognize products of \emph{mixed} states, not just pure states, and that the \textsc{Pure Product State} problems (both one-way LOCC and trace distance versions) are even easier ($\cls{BQP}$-complete).

\section{Conclusion}
\label{sec:conclusion}

We have proved that several separability testing problems are complete for $\cls{BQP}$, $\cls{QMA}$, $\cls{QMA}(2)$, and $\cls{QSZK}$.
These completeness results build upon the work of \cite{HMW13,HMW14}, which exhibits a separability testing problem in $\cls{QIP}(2)$ and another problem complete for $\cls{QIP}$.
The completeness of these problems for a wide range of complexity classes illustrates an important connection between entanglement and quantum computational complexity theory.
In hindsight, it is perhaps natural that these entanglement-related problems capture the expressive power of these classes, since entanglement seems to be the most prominent feature which distinguishes classical from
quantum computational complexity theory.

It is interesting to note the connection between these problems and the differences that give rise to problems complete for different interactive proof classes.
Some patterns emerge: it seems as though mixed state separability requires two messages to be added onto a proof for pure state separability so that the prover may act upon the purification of the mixed state, as is the case for both the ``state'' and ``channel'' versions of these problems.

Two-message quantum interactive proofs continue to be somewhat mysterious.
Extrapolating from our results, the one-way LOCC version of \textsc{Separable State} has the qualities that one would intuitively expect of a $\cls{QIP}(2)$-complete problem.
Despite this intuition, we do not know whether it is $\cls{QIP}(2)$-complete or even $\cls{QMA}$-hard.
However, our work here provides some intuition for why the problem should not be either $\cls{QSZK}$- or $\cls{QMA}$-complete---there are are other problems very different from it that are complete for these classes.

Our work can be extended in a number of directions.
The trace distance version of \textsc{Separable Channel Output} may help to understand the relation between multi-prover quantum interactive proofs with and without entanglement among the provers ($\cls{QMIP}$ versus $\cls{QMIP}^*$).
Similarly, the trace distance version of \textsc{Separable State} may provide further insights.
It would also be worthwhile to characterize the channel version of \textsc{Product State} in order to map out more of the space of separability testing problems.
Such an extension may also help to provide a tighter characterization of classes that rely on ``unentanglement,'' such as $\cls{QMA}(2)$.

It is satisfying that each of the separability testing problems (with the possible exception of the one-way LOCC version of \textsc{Separable State}) is complete for a different complexity class.
Perhaps by studying the remaining related problems and their variants (trace norm versus one-way LOCC norm, separable states versus product states, \emph{etc.}) one may find two different separability testing problems that are nontrivially reducible to each other.

\section*{Acknowledgements}

We thank Claude Cr\'{e}peau, Brian Swingle, and John Watrous for helpful conversations
and the anonymous referees for many helpful suggestions.
Some of this research was conducted while GG was a visitor at the School of Computer Science at McGill University, at which time GG's primary affiliation was the Institute for Quantum Computing and School of Computer Science, University of Waterloo, Waterloo, Ontario, Canada.
MMW began this project while affiliated with the School of Computer Science, McGill University.

\bibliographystyle{tocplain}   
\bibliography{Ref}

\begin{tocauthors}
\begin{tocinfo}[gutoski]
 Gus Gutoski\\
 Postdoctoral Scholar\\
 Perimeter Institute for Theoretical Physics, Waterloo, Ontario, Canada\\
 ggutoski\tocat{}perimeterinstitute\tocdot{}ca\\
 \url{http://www.perimeterinstitute.ca/people/gus-gutoski}
\end{tocinfo}
\begin{tocinfo}[hayden]
  Patrick Hayden\\
  Professor of Physics\\
  Department of Physics, Stanford University, Stanford, California, USA \\
  phayden\tocat{}stanford\tocdot{}edu\\
  \url{http://web.stanford.edu/~phayden}
\end{tocinfo}
\begin{tocinfo}[milner]
  Kevin Milner\\
  PhD student\\
  University of Oxford, Oxford, UK \\
  kamilner\tocat{}kamilner\tocdot{}ca\\
\end{tocinfo}
\begin{tocinfo}[wilde]
  Mark M.~Wilde\\
  Assistant Professor\\
  Hearne Institute for Theoretical Physics, Department of Physics and Astronomy,
  Center for Computation and Technology, Louisiana State University,
  Baton Rouge, Louisiana, USA \\
  mwilde\tocat{}lsu\tocdot{}edu\\
  \url{http://www.markwilde.com}
\end{tocinfo}
\end{tocauthors}

\begin{tocaboutauthors}
\begin{tocabout}[gutoski]
\textsc{Gus Gutoski} was born in Kitchener, Ontario, Canada.
After completing a BMath at the \href{https://uwaterloo.ca/}{University of Waterloo} in Waterloo, Ontario, Canada, he briefly escaped the Region of Waterloo for a MSc at the \href{http://www.ucalgary.ca/}{University of Calgary} in Calgary, Alberta, Canada.
Unfortunately, his graduate advisor, \href{https://cs.uwaterloo.ca/~watrous/}{John Watrous}, subsequently took a position at the University of Waterloo.
Gus reluctantly followed his advisor back to his hometown and completed a PhD in Computer Science at the University of Waterloo.
He is currently a postdoctoral researcher at the \href{http://www.perimeterinstitute.ca/}{Perimeter Institute for Theoretical Physics} in Waterloo and he cannot get enough of Waterloo, Ontario, Canada.
Gus' research interests include quantum complexity theory, quantum cryptography, and, lately, \href{https://bitcoin.org/}{Bitcoin}.
\end{tocabout}
\begin{tocabout}[hayden]
\textsc{Patrick Hayden} received his D.~Phil from \href{www.ox.ac.uk}{University of Oxford} in 2001 and was a postdoc at Caltech before joining the faculty at \href{www.mcgill.ca}{McGill University}, where he spent nine happy years before moving to \href{www.stanford.edu}{Stanford} in 2013. Like many computer scientists, Hayden developed an interest in complexity theory because of its \href{http://arxiv.org/abs/1301.4504}{possible relevance} to black hole physics. His other research interests include skiing and backcountry camping. Outside of work, Hayden enjoys proving theorems in quantum communication theory and studying the emergence of spacetime.
\end{tocabout}
\begin{tocabout}[milner]
\textsc{Kevin Milner} did his undergraduate studies at the \href{www.ualberta.ca}{University of Alberta} where he developed an interest in complexity theory before joining \href{www.mcgill.ca}{McGill University} to study quantum information under the supervision of Patrick Hayden. His hobbies include breaking and fixing the Internet, which he now studies as a D.Phil student supervised by \href{http://www.cs.ox.ac.uk/people/cas.cremers}{Cas Cremers} at the \href{www.ox.ac.uk}{University of Oxford}, and not worrying about the Internet, which he now studies whenever he can.
\end{tocabout}
\begin{tocabout}[wilde]
\textsc{Mark M.~Wilde} was born in
\href{http://en.wikipedia.org/wiki/Metairie,_Louisiana}{Metairie, Louisiana, USA}.
He received the Ph.D. degree in electrical engineering from the
\href{http://www.usc.edu/}{University of Southern California}, Los Angeles, California,
in 2008, and was
advised by \href{http://www-bcf.usc.edu/~tbrun/}{Todd Brun}. He is an
Assistant Professor in the \href{http://www.phys.lsu.edu}{Department of Physics and Astronomy} and the
\href{http://www.cct.lsu.edu/}{Center for Computation and Technology} at
\href{http://www.lsu.edu/}{Louisiana State University}. His
current research interests are in quantum Shannon theory, quantum optical
communication, quantum computational complexity theory, and quantum error
correction. He has elected not to include any humor in his bio because he finds
the above three bios to be about as unfunny as ``unfunny'' can be and
fears that any attempt of his would be worse.
He also recognizes that this is the first time
he has roasted his coauthors in the second-to-last
sentence of a published scientific paper.
\end{tocabout}
\end{tocaboutauthors}

\end{document}